\documentclass{amsart}

\usepackage[latin1]{inputenc}
\usepackage{dsfont}
\usepackage{graphicx}
\usepackage{pstricks}
\usepackage{amssymb}
\usepackage{mathrsfs}
\usepackage{url}

\usepackage{bookmark}

\newtheorem{theorem}{Theorem}
\newtheorem{proposition}{Proposition}
\newtheorem{lemma}[proposition]{Lemma}
\newtheorem{corollary}[proposition]{Corollary}

\theoremstyle{definition}

\theoremstyle{remark}

\newtheorem{remark}[proposition]{Remark}

\numberwithin{equation}{section}
\numberwithin{proposition}{section}

\newcommand\R{{\ensuremath {\mathbb R} }}
\newcommand\C{{\ensuremath {\mathbb C} }}
\newcommand\N{{\ensuremath {\mathbb N} }}
\newcommand\Z{{\ensuremath {\mathbb Z} }}

\newcommand\U{{\ensuremath {\mathbb U} }}
\renewcommand\S{\mathcal{S}}
\renewcommand\phi{\varphi}
\renewcommand\le{\leqslant}
\renewcommand\ge{\geqslant}
\renewcommand\epsilon{\varepsilon}
\renewcommand\hat{\widehat}
\renewcommand\tilde{\widetilde}
\newcommand{\gH}{\mathfrak{H}}
\newcommand{\gK}{\mathfrak{K}}
\newcommand{\gA}{\mathfrak{A}}

\newcommand{\gS}{\mathfrak{S}}

\newcommand{\cP}{\mathcal{P}}

\newcommand\ii{{\ensuremath {\infty}}}

\newcommand\proj[1]{{\ensuremath{| #1\rangle\langle #1|}}}
\newcommand\projneg[1]{{\ensuremath{\chi_{(-\ii,0]}(#1)}}}
\newcommand{\Erbdf}[2]{{\ensuremath{\cE_{\rm{rBDF}}^{#1}(#2)}}}
\newcommand\tro[1]{{\ensuremath{\tr_0(#1)}}}
\renewcommand\d[1]{{\ensuremath{\,\text{d}#1}}}

\newcommand{\cE}{\mathcal{E}}

\newcommand{\F}{\mathcal{F}}

\newcommand{\B}{\mathcal{B}}
\newcommand{\cK}{\mathcal{K}}
\newcommand{\cN}{\mathcal{N}}
\newcommand{\cC}{\mathcal{C}}
\newcommand{\cJ}{\mathcal{J}}

\newcommand{\cF}{\mathcal{F}}

\newcommand{\cV}{\mathcal{V}}

\newcommand{\CJ}{\mathscr{C}}
\newcommand{\alp}{\boldsymbol{\alpha}}

\newcommand{\ind}{\boldsymbol{1}}

\DeclareMathOperator{\supp}{supp}

\DeclareMathOperator{\tr}{Tr}
\DeclareMathOperator{\Det}{Det}
\DeclareMathOperator{\car}{\mathfrak{A}}

\begin{document}
 
\title[Static Pair Creation in Strong Fields]{Static Electron-Positron Pair Creation in \\ Strong Fields for a Nonlinear Dirac model}

\author{Julien Sabin}

\address{D\'epartement de Math\'ematiques, CNRS UMR 8088, Universit\'e de Cergy-Pontoise, 95000 Cergy-Pontoise, France}

\email{julien.sabin@u-cergy.fr}

\date{\today}

\maketitle

\begin{abstract}
 We consider the Hartree-Fock approximation of Quantum Electrodynamics, with the exchange term neglected. We prove that the probability of static electron-positron pair creation for the Dirac vacuum polarized by an external field of strength $Z$ behaves as $1-\exp(-\kappa Z^{2/3})$ for $Z$ large enough. Our method involves two steps. First we estimate the vacuum expectation of general quasi-free states in terms of their total number of particles, which can be of general interest. Then we study the asymptotics of the Hartree-Fock energy when $Z\to+\ii$ which gives the expected bounds. 
\end{abstract}

\tableofcontents

\section*{Introduction}

In 1930, Dirac \cite{Dirac-30} suggested the idea of identifying the vacuum with a sea of virtual electrons with negative kinetic energy. His theory implies that when a sufficiently strong source of energy is provided to the vacuum, some virtual electrons are excited into real electrons, leaving ``holes'' in the Dirac sea. These holes can be interpreted as positrons, the anti-particles of the electrons, which were experimentally observed in 1933 by Anderson \cite{Anderson-33}. The extraction of an electron from the Dirac sea is usually called \emph{electron-positron pair creation}. Sauter \cite{Sauter-31}, and Heisenberg-Euler \cite{HeiEul-36} considered the possibility that an external electromagnetic field could excite the Dirac sea to create those pairs. Schwinger \cite{Schwinger-51a} then computed the probability of dynamical pair creation by a constant, uniform, external electric field in the framework of Quantum Electrodynamics (QED).  The specific phenomenon of pair production triggered by an external, non-quantized field is thus labeled the \emph{Schwinger effect}. It is remarkable that this effect is different from the absorption of photons by the vacuum, which is another possible source for pair creation. The Schwinger effect is based on the fact that the vacuum acts as a polarizable medium which can decay into electron-positron pairs when excited by a sufficiently strong electric field. Although the modern formulation of QED no longer describes the vacuum as a sea of virtual particles, Dirac's theory is still valid in the mean-field approximation \cite{HaiLewSerSol-07}. 

Experimentally, pair creation in electric fields has not been observed yet because it is only non-negligible in a very strong field. However, recent progress in laser physics have permitted to create very strong fields, making the observation of the Schwinger effect possible in the near future \cite{Dunne-09,Tajima-09,Bulanov-10}.

One has to distinguish between \emph{dynamical} and \emph{static} pair creation. Dynamical pair creation consists in studying the time evolution of the vacuum state when an external field is progressively turned on, so that a pair consisting of a scattering electron and a corresponding hole in the Dirac sea is created. The external field is then progressively switched off, and one has to check if the pair still exists when the field is completely turned off. Static pair creation, on the other hand, consists in the study of the absolute ground state (the polarized vacuum) of the Hamiltonian in an external field. Therefore, it is a time-independent process. In this context, the vacuum with an additional particle is energetically more favorable than the vacuum without particle.  Static pair creation is easier to study than dynamical pair creation, but it is also a bit less relevant from the physical point of view. 

When the interactions between particles are neglected (the so-called \emph{linear} case), \emph{static} pair creation was mathematically studied by Klaus and Scharf \cite{KlaSch-77a}. They proved that the probability of pair creation becomes 1 when the strength of the positive external field sufficiently increases such that an eigenvalue of the Hamiltonian of the system crosses zero. In the linear case, \emph{dynamical} pair creation is a very involved phenomenon, whose properties were mathematically understood very recently. Nenciu \cite{Nenciu-80,Nenciu-87} proved that there is a discontinuity in the probability to create pairs as the strength of a specific external field increases, in the adiabatic limit. Later on, Pickl and D\"urr \cite{PicDur-08,PicklPhD} proved that the probability of pair creation tends to 1 in the adiabatic limit, for general over-critical external fields, by carefully studying the resonances created by the eigenvalues diving into the essential spectrum of the Hamiltonian of the system. 
 
This article is devoted to the mathematical study of \emph{static} pair creation in a \emph{nonlinear} model describing the polarized vacuum, taking into account the interactions between particles. This model was first proposed by Chaix and Iracane \cite{ChaIra-89} in 1989 and it has recently been given a solid mathematical ground in a series of papers by Gravejat, Hainzl, Lewin, S\'er\'e, and Solovej \cite{HaiLewSer-05a,HaiLewSer-05b,HaiLewSol-07,HaiLewSer-08,GraLewSer-09}. As in those papers, the main difficulty of our work is the nonlinearity of the model. The more involved study of dynamical pair creation for the same model will be the subject of future work. 

In the considered model, the polarized vacuum in a potential generated by a density of charge $Z\nu$ is described by an operator $P_Z$ on $L^2(\R^3,\C^4)$ (a density matrix). This operator is a solution to the nonlinear equation
\begin{equation*}
 \left\{\begin{array}{ccc}
         P_Z & = & \projneg{D_Z} +\delta \\
	 D_Z & = & D^0 +\alpha\left(\rho_{P_Z-\frac{1}{2}} -Z\nu\right)\star|\cdot|^{-1}
        \end{array}\right. ,
\end{equation*}
where $D^0:=-i\alp\cdot\nabla +\beta$ is the (free) Dirac operator and $\delta$ is any self-adjoint operator such that $0\le\delta\le 1$ with ${\rm rank}(\delta)\subset\ker(D_Z)$. Discarding the operator $\delta$, we see that $P_Z$ is the ground state in the grand canonical ensemble of a dressed Dirac operator with density of charge $Z\nu$ perturbed by the density $\rho_{P_Z-\frac{1}{2}}$ of the vacuum. The polarized vacuum therefore interacts with itself. 

We shall consider the operator $P_Z$ in the limit $Z\to+\ii$. Of our particular interest is the probability that pairs are generated, which is a nonlinear function of $P_Z$ (see Section \ref{sec:def-proba} below). The usual picture \cite{KlaSch-77a,SchSei-82,Hainzl-04} is that if the first eigenvalue $\lambda_1(Z)$ of $D_Z$ is negative, as showed in Figure \ref{fig:spectre}, then the vacuum becomes charged and the probability of creating at least one pair is 1. In the linear case, Hainzl \cite{Hainzl-04} showed that the charge of the vacuum in the external density $Z\nu$ is exactly the number of eigenvalues (counted with multiplicity) of the operator $D^0-tZ\nu\star|\cdot|^{-1}$ crossing 0 when we increase $t$ from $0$ to $1$. However, because of the nonlinearity of the model we study, detecting for which values of $Z$ the first eigenvalue will cross 0 is very difficult. However, the probability of pair creation can be very close to 1 without any crossing, as we will explain in Section \ref{sec:def-proba}.

\begin{figure}[h]
 \input{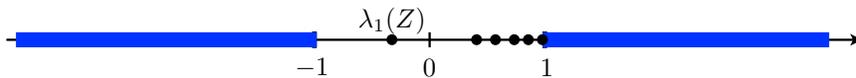}
 \caption{Spectrum of the mean-field one body Hamiltonian $D_Z$.}
\label{fig:spectre}
\end{figure}

More precisely, we prove that the probability of static pair creation behaves as $1-\exp(-\kappa Z^{2/3})$ (see Theorem \ref{th:main}), where $Z$ is the charge of a nucleus put in the vacuum, and $\kappa$ is a constant depending on different parameters of the model such as the cut-off or the shape of the nucleus. The proof relies on the large-$Z$ asymptotics of the polarized vacuum energy, which is obtained by using an appropriate trial state. This implies that the average number of particles of the polarized vacuum is of order $Z^{2/3}$. We then use general estimates showing that the probability to create pairs for a quasi-free quantum state is bigger than $1-\exp(-\kappa' N)$, where $N$ is the average number of particles of the quantum state and $\kappa'$ is a universal constant. Since for the polarized vacuum $N\simeq Z^{2/3}$, the result follows.

The paper is organized as follows. In Section 1 we introduce the Bogoliubov-Dirac-Fock model and we state our main result. In Section 2, we prove the general estimates on quasi-free states on Fock space, which are of independent interest. In the end of Section 2 we come back to our particular setting. In Section 3, we study the large-$Z$ asymptotics of the polarized vacuum energy. Finally, in Section 4 we prove Theorem \ref{th:main} using the tools developed in Section 2 and 3. In Appendix A, we recall some properties of product states, which are used in Section 2.

\bigskip

\noindent\textbf{Acknowledgments.}  I sincerely thank Mathieu Lewin for his precious guidance and constant help. I also acknowledge support from the ERC MNIQS-258023 and from the ANR ``NoNAP'' (ANR-10-BLAN 0101) of the French ministry of research.

\section{Estimate on the probability to create pairs}

\subsection{Probability to create a pair}\label{sec:def-proba}

The first quantity to define is the probability to create a pair. Let $\gH_+,\gH_-$ be (separable) Hilbert spaces, representing the one particle (resp. anti-particle) space. The natural space to describe a system with an arbitrary number of particles/anti-particles is the Fock space 
$$\cF_0:=\cF(\gH_+)\otimes\cF(\gH_-),$$
with the usual notation $\cF(\gH):=\oplus_{N\ge0}\wedge_1^N\gH$ for any Hilbert space $\gH$ and with the convention $\wedge_1^0\gH:=\C$. We also define the \emph{vacuum state} $\Omega:=\Omega_+\otimes\Omega_-\in\F_0$ where $\Omega_\pm:=1\oplus0\oplus0\cdots\in\cF(\gH_\pm)$. Recall that a state over $\cF_0$ can be defined \footnote{For convenience, we will later define a state as a linear form on the tensor product of CAR algebras $\car[\gH_+]\otimes\car[\gH_-]\subset\B(\cF_0)$, which in our context does not change anything since all the states we consider are normal.} as a positive linear functional $\omega:\B(\cF_0)\to\C$ with $\omega(\rm{Id}_{\cF_0})=1$, where $\B(\cF_0)$ is the set of all bounded linear operators on $\cF_0$. Notice that any normalized $\psi\in\cF_0$ defines a state $\omega_\psi$ (called \emph{pure} state) by the formula $\omega_\psi(A)=\langle\psi,A\psi\rangle_{\cF_0}$, where $\langle\cdot,\cdot\rangle_{\cF_0}$ is the usual inner product on $\cF_0$. Following \cite[Corollary 4.1]{PicklPhD}, \cite[Eq. (10.154)]{Thaller}, and \cite[Section 2]{Nenciu-87}, we define the probability $p(\omega)$ for a state $\omega$ to create a particle/anti-particle pair by
\begin{equation}
 \boxed{
 p(\omega):=1-\omega\left(\proj{\Omega}\right),
}
\end{equation}
where $\proj{\Omega}\in\B(\cF_0)$ is the orthogonal projection on $\C\Omega$. For a pure state $\psi=\psi_{0,0}\oplus\psi_{0,1}\oplus\psi_{1,0}\oplus\cdots\in\cF_0$, we have $p(\omega_\psi)=1-\left|\psi_{0,0}\right|^2$.Therefore, $p(\omega_\psi)=0$ if and only if $\psi=\Omega$ (the vacuum has probability zero to create pairs), while $p(\omega_\psi)=1$ if and only if $\psi_{0,0}=0$. In the latter case, notice that $\psi$ does not litterally contain pairs, in the sense that its number of particles may not be equal to its number of anti-particles. This definition merely measures the probability that a state contains real particles/anti-particles.

Typically, $\Omega$ represents the free (or bare) vacuum and we want to measure the probability of a perturbation $\Omega'$ of $\Omega$, representing the polarized (or dressed) vacuum in the presence of an external electric field, to have pairs. Assuming that $\Omega'$ is a pure quasi-free state, we have the well-known formula (see e.g. \cite[Theorem 10.6]{Thaller}, \cite[Theorem 2.2]{BacLieSol-94}, or \cite[Theorem 5]{HaiLewSer-08})
\begin{equation}\label{formula-vacuum}
\Omega'=\prod_{i\ge1}\frac{1}{\sqrt{1+\lambda_i^2}}\prod_{n=1}^N a_0^*(f_n)\prod_{m=1}^M b_0^*(g_m)\prod_{i\ge 1}\left(1+\lambda_ia_0^*(v_i)b^*_0(u_i)\right)\Omega,
\end{equation}
where $a_0^*$ (resp. $b_0^*$)  is the free particle (resp. anti-particle) creation operator, $(f_n)_n\cup(v_i)_i$ (resp. $(g_m)_m\cup(u_i)_i$) is an orthonormal set for $\gH_+$ (resp. $\gH_-$), and $(\lambda_i)_i\in\ell^2(\R_+)$.

 From the formula \eqref{formula-vacuum}, we see that $p(\omega_{\Omega'})=1$ as soon as $N>0$ or $M>0$. Moreover, in this case real particles in the states $(f_n)_n$ and real anti-particles in the states $(g_m)_m$ have been created. In the linear case, Klaus and Scharf \cite{KlaSch-77a} proved that $N,M\neq0$ if the external field is strong enough. However, there can be a high probability to create pairs even when $N=M=0$. Indeed, since in this case we have
 $$|\langle\Omega',\Omega\rangle_{\cF_0}|^2=\prod_{i\ge1}\frac{1}{1+\lambda_i^2},$$ 
 one sees that $p(\omega_{\Omega'})$ is close to 1 when the $\lambda_i$ are large enough. One simple condition is that $\sum_i\lambda_i^2(1+\lambda_i^2)^{-1}$ is large enough, by the inequality 
$$\prod_{i\ge1}\frac{1}{1+\lambda_i^2}\le\exp\left[-\sum_i\frac{\lambda_i^2}{1+\lambda_i^2}\right].$$
Note that this is indeed (half) the average total number of particle of the state $\Omega'$ (number of particle + number of anti-particle),
$$\omega_{\Omega'}(\cN)=2\sum_i\frac{\lambda_i^2}{1+\lambda_i^2},$$
where $\cN$ is the usual number operator on $\cF_0$ (see formula \eqref{eq:numberop}). Hence $p(\omega_{\Omega'})$ is close to 1 when $\omega_{\Omega'}(\cN)$ is large enough. While the non-vanishing $N,M$ case can be interpreted as the creation of real particles, this second explanation for an increasing $p(\omega_{\Omega'})$ can be interpreted as a ``virtual pair creation''.  In this article, we study an analog of ``virtual pair creation'' for more general states than those given by formula \eqref{formula-vacuum}.

\subsection{Static pair creation in the reduced BDF approximation}\label{sec:BDF-notations}

For noninteracting electrons in an external field $V$, the polarized vacuum $\Omega'$ is the unique Hartree-Fock state whose density matrix is \cite{KlaSch-77a,Hainzl-04}
$$P=\projneg{D^0+V}.$$
 In this article, we will rather use the reduced Bogoliubov-Dirac-Fock approximation, a non-linear model enabling to describe an interacting vacuum in which $V$ is a function of $P$ itself. It was introduced by Hainzl, Lewin, S\'er\'e and Solovej in a series of articles \cite{HaiLewSer-05a,HaiLewSer-05b,HaiLewSol-07,HaiLewSer-08} after the pioneering work of Chaix, Iracane, and Lions \cite{ChaIra-89,ChaIraLio-89}. We will now briefly recall the model and the results needed for our study. 

In units where $m=c=\hbar=1$, the reduced Bogoliubov-Dirac-Fock (rBDF) energy functional is the (formal) difference between the energy of the state $P$ and that of the free vacuum $P^0_-=\projneg{D^0}$, with the exchange term dropped. It depends only on the variable $Q=P-P^0_-$, 
\begin{equation}\label{rBDF-energy}
 \Erbdf{\nu}{Q}:=\tro{D^0Q}-\alpha D(\rho_Q,\nu)+\frac{\alpha}{2}D(\rho_Q,\rho_Q).
\end{equation}
Here,  $\alpha>0$ is the coupling constant and $\nu:\R^3\to\R$ is the external charge density belonging to the Coulomb space 
$$\cC:=\left\{f\in\mathscr{S}'(\R^3):\:\int_{\R^3}\frac{|\widehat{f}(k)|^2}{|k|^2}\d{k}<+\ii\right\},$$ 
endowed with the inner product $D(\rho_1,\rho_2)=\int|k|^{-2}\widehat{\rho_1}(k)\overline{\widehat{\rho_2}(k)}$ (the hat denotes the Fourier transform\footnote{Our convention is $\hat{f}(k)=(2\pi)^{-3/2}\int_{\R^3}f(x)e^{-ik\cdot x}\d{x}, \quad\forall k\in\R^3$.}). We also use the notation $\|\rho\|_\cC=D(\rho,\rho)^{1/2}$ for any $\rho\in\cC$. In order to define the domain of the rBDF energy functional, let us fix a cut-off $\Lambda>0$ and define the one-particle Hilbert space
$$\gH_\Lambda:=\left\{f\in L^2(\R^3,\C^4),\quad\supp\widehat{f}\subset B(0,\Lambda)\right\}.$$
The operator $D^0=-i\alp\cdot\nabla+\beta$ is the usual Dirac operator on $L^2(\R^3,\C^4)$, where $\alpha_1,\alpha_2,\alpha_3,\beta$ are the Dirac matrices acting on $\C^4$,
$$\alpha_i=\left(\begin{array}{cc}
   0 & \sigma_i\\
  \sigma_i & 0
  \end{array}\right),\quad i=1,2,3, \qquad
\beta=\left(\begin{array}{cc}
   \text{Id}_{\C^2} & 0\\
  0 & -\text{Id}_{\C^2}
  \end{array}\right),
$$
and $(\sigma_i)_{i=1,2,3}$ are the Pauli matrices,
$$\sigma_1=\left(\begin{array}{cc}
   0 & 1\\
   1 & 0
  \end{array}\right),\qquad
  \sigma_2=\left(\begin{array}{cc}
   0 & -i\\
   i & 0
  \end{array}\right),\qquad
  \sigma_3=\left(\begin{array}{cc}
   1 & 0\\
   0 & -1
  \end{array}\right).
$$

 The operator $D^0$ stabilizes $\gH_\Lambda$, and its restriction to $\gH_\Lambda$ defines a bounded operator on $\gH_\Lambda$, which we still denote by $D^0$. For convenience we introduce $P^0_+:=1-P^0_-$.
We denote by $\gS_p(\gH)$ the Schatten class of all bounded operators $A$ on the Hilbert space $\gH$ such that $\tr(|A|^p)<\ii$. For any operator $Q$ on $\gH_\Lambda$ and for any $\epsilon,\epsilon'\in\{+,-\}$, we let $Q_{\epsilon\epsilon'}:=P^0_\epsilon QP^0_{\epsilon'}$ and we define 
$$\gS_{1,{P^0_-}}(\gH_\Lambda):=\left\{ Q\in\gS_2(\gH_\Lambda),\quad Q_{++},Q_{--}\in\gS_1(\gH_\Lambda)\right\}.$$
It is a Banach space endowed with the norm
$$\|Q\|_{1,P^0_-}:=\|Q_{++}\|_{\gS_1}+\|Q_{--}\|_{\gS_1}+\|Q_{+-}\|_{\gS_2}+\|Q_{-+}\|_{\gS_2}.$$

For any $Q\in\gS_{1,{P^0_-}}(\gH_\Lambda)$, we define its generalized trace by
$$\tro{Q}:=\tr\left(Q_{++}+Q_{--}\right),$$
and its density $\rho_Q$ by $\rho_Q(x):=\tr_{\C^4}(Q(x,x))$ for all $x\in\R^3$, where $Q(x,y)$ denotes the $4\times4$ matrix kernel of $Q$. This density $\rho_Q$ is well defined since $\supp\hat{Q}(\cdot,\cdot)\subset B(0,\Lambda)\times B(0,\Lambda)$ implies that $Q(\cdot,\cdot)$ is smooth. Furthermore, it is proved in \cite[Lemma 1]{HaiLewSer-08} that $\rho_Q\in L^2(\R^3)\cap\cC$ for any $Q\in\gS_{1,{P^0_-}}(\gH_\Lambda)$. We conclude that the rBDF energy functional is well-defined on the convex set
\begin{equation}\label{eq:cK}
\cK:=\left\{Q\in\gS_{1,{P^0_-}}(\gH_\Lambda),\quad Q=Q^*,\quad -P^0_-\le Q\le 1-P^0_-\right\}. 
\end{equation}
Notice that the kinetic part of the rBDF energy is well defined since 
$$\tro{D^0Q}=\tr(|D^0|(Q_{++}-Q_{--})).$$
The variational set $\cK$ is the convex hull of $\{P-P^0_-,\,P=P^2=P^*,\,P-P^0_-\in\gS_2(\gH_\Lambda)\}$, where $P$ is the density matrix of a pure Hartree-Fock state, which is a Hilbert-Schmidt perturbation of the free vacuum $P^0_-$. The rigorous derivation of the rBDF energy functional and the motivation for this functional setting can be found in \cite{HaiLewSol-07,HaiLewSer-05a}. 

For any $Z>0$ and $\nu\in\cC$, the rBDF energy functional $\cE_{\text{rBDF}}^{Z\nu}$ admits global minimizers $Q_Z$ on $\cK$. Minimizers are not necessarily unique, but they always share the same density $\rho_Z:=\rho_{Q_Z}$. Any minimizer $Q_Z$  satisfies the self-consistent equation
\begin{equation}
 \left\{\begin{array}{ccc}
         Q_Z & = & \projneg{D_Z}-P^0_- +\delta \\
	 D_Z & = & D^0 +\alpha\left(\rho_Z-Z\nu\right)\star|\cdot|^{-1}
        \end{array}\right. ,
\end{equation}
where $\delta$ is a self-adjoint operator such that $0\le\delta\le 1$ and ${\rm rank}(\delta)\subset\ker(D_Z)$. Hence, uniqueness holds if and only if $\ker(D_Z)=\{0\}$. Notice that since the density $\rho_Z$ is unique, the operator $D_Z$ is itself unique. Any minimizer of $\cE_{\text{rBDF}}^{Z\nu}$ on $\cK$ is interpreted as a generalized one-particle density matrix of a BDF state $\omega_{\text{vac}}^{Z\nu}$ (see Section \ref{sec:BDFstates}) representing the \emph{polarized vacuum} in the potential $Z\nu\star|\cdot|^{-1}$.  When there is a unique minimizer $Q_Z$, it is a difference of two projectors and hence it is the generalized one-particle density matrix of a pure state. We emphasize that there is no charge constraint in this minimization problem: In fact, the polarized vacuum could (and should) be charged when $Z$ is very large. In this case, one may think that an electron-positron pair is created, with the positron sent to infinity. 

We want to estimate $p(\omega_{\text{vac}}^{Z\nu})$ in terms of $Z$, and confirm the picture that the stronger the field, the more pairs are created. As a consequence, we will fix a non-zero density $\nu$ (interpreted as the shape of the external charge density) and study $p_Z:=p(\omega_{\text{vac}}^{Z\nu})$ for $Z>0$ large. Our main result is the following. 

\begin{theorem}\label{th:main}
 Let $\alpha>0$ and $\Lambda>0$. Let $\nu\in\cC$ such that $\int_{\R^3}(1+|x|)|\nu(x)|\d{x}<\ii$ and $q:=\int_{\R^3}\nu\neq0$. Then, there exists a constant $Z_1>0$ and a constant $\kappa>0$ such that for all $Z> Z_1$ we have 
 \begin{equation}\label{est:main}
  \boxed{
 p_Z\ge 1-e^{-\kappa Z^{2/3}}.
 }
 \end{equation}
 The constant $Z_1$ is equal to $\tilde{Z_1}/\int_{\R^3}|\nu|$ where $\tilde{Z_1}$ is defined in Equation \eqref{eq:Z_1} and $\tilde{Z_1}$ depends only on $\Lambda$,  $\alpha$, $|q|$,  and $\int_{\R^3}|x||\nu|$. The constant $\kappa$ equals $0.0941248\ldots\times C(\int_{\R^3}|\nu|)^{2/3}$, where $C$ is defined in Equation \eqref{eq:cste} below.
\end{theorem}

\begin{remark}
 The assumption $|x|\nu\in L^1(\R^3)$ allows us to have an explicit estimate. If we remove this assumption, we can still prove the weaker result that $p_Z\to1$ as $Z\to\ii$. In the sequel, by rescaling $Z$ if needed, we will assume that 
 $$
\boxed{
\int_{\R^3}|\nu|=1.
}
$$ 
If $\int\nu=0$, we expect an asymptotics lower than $Z^{2/3}$, but we are unable to prove it. 
\end{remark}

\begin{remark}
 We will see that $Z_1\sim\text{const.}\times\alpha^{-3/2}$ as $\alpha\to0$.
\end{remark}

Theorem \ref{th:main} says that in a very strong field, $Z\gg1$, the probability to create at least one electron-positron pair is very close to 1. It is reasonable to think that for some sufficiently large $Z$, the first eigenvalue of $D_Z$ crosses 0 in which case $p_Z=1$. However, determining the behaviour of the eigenvalue of $D_Z$ as $Z$ increases is difficult because of the nonlinearity of the model and because we are in a regime far from being perturbative. For all these reasons, the estimate \eqref{est:main} on $p_Z$ is the best we have so far. For very large $Z$, one expects that many electron-positron pairs will be generated. We conjecture that we have indeed $\omega_{\text{vac}}^{Z\nu}(\cP_k)\to0$ as $Z\to+\ii$ for all $k\in\N$, where $\cP_k$ is the orthogonal projector on the $k$-particle space in Fock space (see Section \ref{sec:quasi-free}). This would mean that for large $Z$, the probability to create at least $k$-pairs is very close to 1. Our method of proof only gives this result for $k=0$. However, if we assume that $\omega_{\text{vac}}^{Z\nu}$ is a \emph{pure} quasi-free state for all $Z$ large enough (which is the case if $0\notin\sigma(D_Z)$), then the conjecture follows from Proposition \ref{prop:HFBpure} below.

\subsection{Strategy of the proof}

The proof is separated into two parts. The first one consists in estimating the energy of the polarized vacuum,
\begin{equation}
 E(Z):=\Erbdf{Z\nu}{Q_{\text{vac}}^{Z\nu}}=\inf\left\{\Erbdf{Z\nu}{Q},\quad Q\in\cK\right\},
\end{equation}
from above by $- cZ^{5/3}$. We will also give a lower bound $E(Z)\gtrsim -Z^{5/3}$ to show that the power $5/3$ is optimal, although we only need the upper bound for the proof of Theorem \ref{th:main}. From this estimate, we then infer that the average number of particles (counted relatively to that of the free vacuum, see Section \ref{sec:BDFstates}) in the polarized vacuum satisfies
$$\tr\left((Q_{\text{vac}}^{Z\nu})_{++}-(Q_{\text{vac}}^{Z\nu})_{--}\right)\gtrsim Z^{2/3}.$$
The precise statements of these results and their proofs can be found in Section \ref{sec:asymptotics}. 

In a second part, we prove an estimate on the vacuum expectation $\omega(\proj{\Omega})$ for a quasi-free state $\omega$, in terms of its average number of particle $\omega(\cN)$. These estimates are of independent interest and therefore we also provide several other estimates for the distribution of quasi-free states in the $k$-particle spaces. These results are contained in Section \ref{sec:quasi-free}. Finally, in Section \ref{sec:proof}, we combine the two parts and prove Theorem \ref{th:main}.

\section{On the distribution of quasi-free states in the $k$-particle spaces}\label{sec:quasi-free}

In this section, we consider general quasi-free states. Only in Section \ref{sec:BDFstates} we come back to our particular situation of pair creation. We start by introducing the notation used throughout this section.

\subsection{Notation}\label{sec:notation}
Let  $(\gH,\langle\cdot,\cdot\rangle)$ be a complex, separable Hilbert space whose inner product $\langle\cdot,\cdot\rangle$ is linear in the second argument. We also need an anti-linear operator ${J:\gH\to\gK}$ such that $J^*J=\text{Id}_\gH$, where $(\gK,\langle\cdot,\cdot\rangle_\gK)$ is another complex Hilbert space\footnote{Recall that the adjoint $J^*$ of an anti-linear operator is defined as $\langle J^*f,g\rangle:=\langle J g,f\rangle_\gK$ for all $f\in\gK$ and $g\in\gH$. Typically, one chooses $\gK=\gH$ and $J$ the complex conjugation, or $\gK=\gH^*$ and $J(f)=\langle f,\cdot\rangle$ \cite{Solovej-07}. Here we keep $\gK$ abstract because this will be useful for the construction of BDF states in Section \ref{sec:BDFstates}.}. Let $\F:=\oplus_{N\ge 0}\gH^N$ be the associated Fock space with $\gH^N:=\wedge_1^N\gH$. We still denote by $\Omega=1\oplus0\oplus\cdots\in\cF$ the vacuum vector. For $k\in\N$, we denote by $\cP_k\in\B(\cF)$ the orthogonal projection on $\gH^k\subset\cF$. We recall from Section \ref{sec:def-proba} that $\B(\cF)$ is the space of all linear bounded operators on $\cF$.  Let $\car[\gH]$ be the CAR unital $C^*$-subalgebra of $\B(\cF)$ generated by the usual creation (resp. annihilation operators) $a^*(f)$ (resp. $a(f)$), for $f\in\gH$. We denote by $\cN$ the particle number operator on $\F$,
$$\cN:=\bigoplus_{k\ge 0}k\,\text{Id}_{\gH^k}=\sum_{i\ge0}a^*(f_i)a(f_i)$$
for any orthonormal basis $(f_i)_{i\ge0}$ in $\gH$. Then $\cP_k=\ind_{\{\cN=k\}}$ for all $k\in\N$. A state on $\car[\gH]$ is a non-negative linear functional $\omega:\car[\gH]\to\C$ which is normalized: $\omega(\text{Id}_\cF)=1$. A state $\omega$ is called \emph{normal} if there exists a non-negative operator $G$ on $\F$ (sometimes called the density matrix of $\omega$) such that $\tr_\cF(G)=1$ and $\omega(A)=\tr_\cF(GA)$ for all $A\in\car[\gH]$. Of particular interest are the \emph{pure} states which are normal states with $G=\proj{\psi}$ for $\psi\in\cF$ with $\|\psi\|_\cF=1$. We define the average particle number of $\omega$ as $$\omega(\cN):=\sum_{i\ge0}\omega(a^*(f_i)a(f_i))\in[0,+\ii].$$ The \emph{one-particle density matrix} (1-pdm) $\gamma$ of $\omega$ is the operator defined by
$$\langle g,\gamma f\rangle:=\omega(a^*(f)a(g)),$$
for all $f,g\in\gH$. It is a self-adjoint operator on $\gH$, satisfying $0\le\gamma\le 1$. In the same fashion, we define its \emph{pairing matrix} $\alpha:\gK\to\gH$ which is a linear operator on $\gK$ by
$$\langle \alpha J f,g\rangle:=\omega(a^*(f)a^*(g)),$$
for all $f,g\in\gH$. It satisfies $(\alpha J)^*=-\alpha J$. Moreover, if we define the operator $\Gamma(\gamma,\alpha)$ on $\gH\oplus\gK$ by block
\begin{equation}\label{Gamma}
 \Gamma(\gamma,\alpha):=\left(
\begin{array}{cc}
 \gamma & \alpha \\
 \alpha^* & 1-J\gamma J^*
\end{array}
\right),
\end{equation}
then $0\le\Gamma(\gamma,\alpha)\le 1$, see \cite[Lemma 2.1]{BacLieSol-94}. This last relation implies that 
\begin{equation}\label{ineq:gammalpha}
\gamma^2+\alpha\alpha^*\le\gamma,
\end{equation}
 in the sense of quadratic forms on $\gH$. Notice also that $\omega(\cN)=\tr(\gamma)$. A state $\omega$ is called \emph{quasi-free} if for any operators $e_1,\ldots,e_{2p}$ which are either a $a^*(f)$ or a $a(g)$ for any $f,g\in\gH$, then $\omega(e_1 e_2\ldots e_{2p-1})=0$ for any $p\ge 1$ and 
\begin{equation}\label{wick}
 \omega(e_1 e_2\ldots e_{2p})=\sum_{\pi\in\widetilde{\S_{2p}}} (-1)^{\epsilon(\pi)}\omega(e_{\pi(1)}e_{\pi(2)})\ldots\omega(e_{\pi(2p-1)}e_{\pi(2p)}),
\end{equation}
where $\widetilde{\S_{2p}}$ is the set of permutations of $\{1,\ldots,2p\}$ which verify $\pi(1)<\pi(3)<\cdots<\pi(2p-1)$ and $\pi(2j-1)<\pi(2j)$ for all $1\leqslant j \leqslant p$, and $\epsilon(\pi)$ is the parity of the permutation $\pi$. The relation (\ref{wick}) is called the Wick formula. From this definition, we see that a quasi-free state is completely determined by its density matrices $(\gamma,\alpha)$. We  recall \cite[Theorem 2.3]{BacLieSol-94}
\begin{proposition}\label{BacLieSol}
 For any $(\gamma,\alpha)$ such that $0\le\Gamma(\gamma,\alpha)\le1$ with additionally $\tr(\gamma)<+\ii$, there exists a unique quasi-free state $\omega$ on $\car[\gH]$ with finite number of particle such that $\gamma$ is its 1-particle density matrix and $\alpha$ its pairing matrix. Furthermore, $\omega$ is normal: there exists $G:\cF\to\cF$ with $0\le G\le 1$ and $\tr_\cF(G)=1$ such that $\omega(A)=\tr_\cF(GA)$ for all $A\in\car[\gH]$.
\end{proposition}

Now, we need some terminology, which is not universal in the literature. We call \emph{Hartree-Fock} (HF) states the quasi-free states with $\tr(\gamma)<+\ii$ and $\alpha=0$, because when such states are pure, they are usual Slater determinants. Quasi-free states with $\tr(\gamma)<+\ii$ and $\alpha\neq0$, are called \emph{Hartree-Fock-Bogoliubov} (HFB) states. Pure HFB states are particularly simple since they are Bogoliubov rotations of the vacuum $\Omega$.  The aim of this section is to study the distribution of quasi-free states in the particle subspaces $\gH^k$, in terms of $\tr(\gamma)$. Our results are different for HF or HFB, pure or mixed states. 

\subsection{Motivation} Quasi-free states are also called \emph{Gaussian} states, in particular because they can be written as (limits of) Gibbs states of quadratic Hamiltonians (i.e. normal states with density matrices $e^{-\beta\mathbb{H}}/\tr(e^{-\beta\mathbb{H}})$, where $\mathbb{H}$ is a quadratic Hamiltonian).  The Gaussian character of quasi-free states is however deeper. In this section we will show that the distribution of a quasi-free state $\omega$ over the different $\gH^k$, that is $(\omega(\cP_k))_{k\ge0}$, also has some Gaussian characteristics. More precisely, we will provide estimates of the form 
$$\omega(\cP_k)\le c_k e^{-c'\omega(\cN)}.$$ 
This estimate means that a quasi-free state which has a large average number of particles $\omega(\cN)\gg1$ necessarily has an exponentially small vacuum expectation 
$$\omega(\proj{\Omega})\le c_0 e^{-c'\omega(\cN)}.$$

Let us explain the picture in a commutative setting. Let $f(x)=\pi^{-1/2}e^{-|x-a|^2}$ for $a\in\R_+$ be a Gaussian function such that $\int_\R f=1$. Then $a=\int_\R xf(x)\d{x}$ is the average position of $f$, as desbribed in Figure \ref{fig:gaussian}. Now if $a$ goes to $+\ii$, the whole function moves to infinity and in particular $f(0)$ becomes smaller and smaller: $f(0)=\pi^{-1/2}e^{-a^2}$. In other words, as the average position of $f$ goes to infinity, $f(0)$ goes to zero (and this is true for any $f(x_0)$ with $x_0$ fixed). We will prove a similar fact for quasi-free states. Indeed, for any quasi-free state $\omega$ we have $1=\omega(\text{Id}_\cF)=\sum_{k\ge0}\omega(\cP_k)$, which is the analog of $\int_\R f(x)\d{x}=1$. We also know that $\omega(\cN)=\sum_{k\ge0}k\omega(\cP_k)$ is the average number of particle of $\omega$; it is the analog of $\int_\R xf(x)\d{x}$. We want to prove that when $\omega(\cN)$ is large, then the main part of $\omega$ lives in the high-$k$ particle spaces, that is $\omega$ ``follows'' its average number of particles, as shown in Figure \ref{fig:gaussian}. The analog of $f(0)$ in this case is $\omega(\proj{\Omega})$ and we thus want to prove that $\omega(\proj{\Omega})$ goes to zero as $\omega(\cN)$ goes to $+\ii$, for any quasi-free state $\omega$. A natural extension of this result would be that $\omega(\cP_{k_0})$ also goes to zero for any fixed $k_0$. 

\begin{figure}[h]
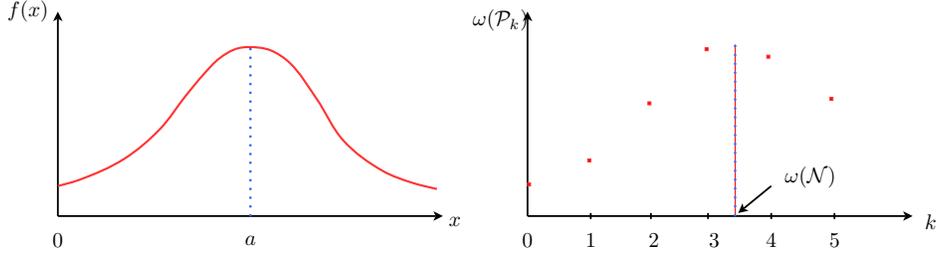

\begin{center}
\hspace{-0.075\linewidth} 
\begin{minipage}{0.45\linewidth}
\centering \scalebox{0.8}{ \input{Gaussienne.pst} }
\end{minipage}
\hspace{0.025\linewidth}
\begin{minipage}{0.45\linewidth}
 \centering \scalebox{0.8}{ \input{quasifree.pst} }
\end{minipage}
\end{center}
\caption{Analogy between a Gaussian function and a quasi-free state}
\label{fig:gaussian}
\end{figure}

We will provide explicit estimates depending on the properties of the quasi-free state (pure, mixed, HF or HFB). In the most general case of mixed HFB states, we only derive a bound on the vacuum expectation. The following table tells us where each case is treated. 

\begin{center}
\begin{tabular}{|c|c|c|}
 \hline
  & Pure  & Mixed \\
\hline
HF ($\alpha=0$) & \multicolumn{2}{c|}{Section \ref{sec:HF}} \\
\hline
HFB ($\alpha\neq0$) & Section \ref{sec:HFBpure} & Section \ref{sec:HFBmixed}\\
\hline
\end{tabular}
\end{center}

In spite of their usefulness, we have not found the following estimates in the literature. One main reason is probably that $e^{-\beta\cN}$ does \emph{not} belong to the CAR algebra, hence $\omega(e^{-\beta\cN})$ only makes sense for normal states.

\begin{remark}
 A useful tool for the proofs of the following results is the notion of \emph{product state}. A product state $\otimes_i\omega_i$ is a state on $\cF(\oplus_i\gH_i)\simeq\otimes_i\cF(\gH_i)$ when each $\omega_i$ is a state on $\cF(\gH_i)$. While this notion is intuitive, we recall how it is precisely defined in Appendix \ref{prodstates}.
\end{remark}

\subsection{Hartree-Fock case}\label{sec:HF}

\begin{proposition}[HF case]\label{prop:HF}
 Let $\omega$ be a quasi-free state with $\tr(\gamma)<+\ii$ and $\alpha=0$. Then for any $\beta\ge0$ we have 
 \begin{equation}\label{eq:HFequality}
  \boxed{
 \omega\left(e^{-\beta\cN}\right)=\Det_\gH\left(1+(e^{-\beta}-1)\gamma\right).
 }
 \end{equation}
 We also have the following estimate 
 \begin{equation}\label{eq:HFestimate}
  \boxed{
 \omega(\cP_k)\le\frac{(e\tr(\gamma))^k}{k!}e^{-\tr(\gamma)},
 }
 \end{equation}
 for all $k\ge k_0$ while $\omega(\cP_k)=0$ if $k<k_0$, where $k_0:=\dim\ker(\gamma-1)$.
\end{proposition}
\begin{remark}\label{rk:1}
 This estimate implies that for any fixed $k\in\N$, $\omega(\cP_k)\to0$ as $\tr(\gamma)\to+\ii$, which is the expected behaviour.
\end{remark}
\begin{remark}\label{rk:2}
 A more theoretical corollary of \eqref{eq:HFequality} is that for any fixed $f:\N\to\C$ vanishing at infinity, $\omega(f(\cN))$ goes to zero as $\tr(\gamma)$ goes to $+\ii$. This uses the fact that the algebra of these $f$s is generated by the $(e^{-\beta\cdot})_{\beta\ge0}$. In the same fashion, one can also prove that $\omega(K)\to0$ as $\tr(\gamma)\to+\ii$, for any fixed compact operator $K$.
\end{remark}
\begin{proof}
 Since $\gamma$ is trace-class it can be diagonalized in an orthonormal basis $(f_i)_{i\in\N}$, $\gamma=\sum_{i\ge0}\lambda_i\proj{f_i}$. For all $i$, let $\omega_i$ be the unique quasi-free state on $\F(\C f_i)$ having $\lambda_i\text{Id}_{\C f_i}$ as 1-pdm and $0_{\C f_i}$ as its pairing matrix. Then by Proposition \ref{propprodstates} in Appendix \ref{prodstates}, one has $\omega=\otimes_i\omega_i$. Moreover, since
 $$Te^{-\beta\cN}T^*=\bigotimes_{i\in\N}\left( 1+(e^{-\beta}-1)a^*(f_i)a(f_i)\right),$$
 where $T$ is the isometry between $\cF(\oplus_i\gH_i)$ and $\otimes_i\cF(\gH_i)$ defined in Appendix \ref{prodstates}, we have 
 \begin{eqnarray*}
  \omega(e^{-\beta\cN}) & = & \prod_{i\in\N}\omega_i\left( 1+(e^{-\beta}-1)a^*(f_i)a(f_i)\right)=\prod_{i\in\N}\left(1+(e^{-\beta}-1)\lambda_i\right)\\
 & = & \Det_\gH\left(1+(e^{-\beta}-1)\gamma\right).
 \end{eqnarray*}
 To prove \eqref{eq:HFestimate}, we notice that for all $\beta\ge0$, $\omega(e^{-\beta\cN})=\sum_{k\ge0}e^{-\beta k}\omega(\cP_k)$, and we identify the coefficients of $e^{-\beta k}$ in $\prod_{i\in\N}\left(1+(e^{-\beta}-1)\lambda_i\right)$. This yields
 \begin{eqnarray*}
  \omega(\cP_k) & = & \sum_{\substack{I\subset\N \\ \#I=k}}\prod_{i\in I}\lambda_i\prod_{j\notin I}(1-\lambda_j)\\
 & \le & \sum_{\substack{I\subset\N \\ \#I=k}}\left(\prod_{i\in I}\lambda_i\right)e^{-\sum_{j\notin I}\lambda_j} \\
 & \le & e^{-\tr(\gamma)}\sum_{\substack{I\subset\N \\ \#I=k}}\prod_{i\in I}\lambda_i e^{\lambda_i} \\
 & \le & \frac{e^k}{k!}e^{-\tr(\gamma)}\sum_{i_1,\ldots,i_k\in\N}\lambda_{i_1}\cdots\lambda_{i_k}=\frac{(e\tr(\gamma))^k}{k!}e^{-\tr(\gamma)},
 \end{eqnarray*}
 where we used that $0\le\lambda_i\le1$ for all $i$. Notice from the first equality that $\omega(\cP_k)=0$ if $k<\dim\ker(\gamma-1)$.
\end{proof}


\subsection{Pure Hartree-Fock-Bogoliubov case}\label{sec:HFBpure} 

\begin{proposition}[Pure HFB case]\label{prop:HFBpure}
  Let $\omega$ a quasi-free \emph{pure} state with $\tr(\gamma)<+\ii$. Then for any $\beta\ge 0$ we have 
 \begin{equation}\label{eq:HFBpureequality}
  \boxed{
  \omega(e^{-\beta\cN})=\Det_\gH\sqrt{1+(e^{-2\beta}-1)\gamma}.
  }
 \end{equation}
 We also have the following estimate for all $k=k_0+2\ell$ with $k_0:=\dim\ker(\gamma-1)$ and $\ell\ge0$
 \begin{equation}
  \boxed{
 \omega(\cP_k)\le\frac{e^{k/2}}{\ell!}\left(\frac{\tr(\gamma)}{2}\right)^\ell e^{-\frac{\tr(\gamma)}{2}},
 }
 \end{equation}
 while $\omega(\cP_k)=0$ if $k<k_0$ or $k=k_0+2\ell+1$.
\end{proposition}

\begin{proof}
 It is well-known \cite[Theorem 2.6]{BacLieSol-94} that $\omega$ is pure if and only if $\Gamma(\gamma,\alpha)^2=\Gamma(\gamma,\alpha)$, which is equivalent to $\gamma^2+\alpha\alpha^*=\gamma$ and $[\gamma,\alpha J]=0$. The operator $\gamma$ is trace-class and $\alpha J$ is anti-hermitian and Hilbert-Schmidt, hence both $\gamma$ and $\alpha J$ are diagonalizable. Since they commute, they are simultaneously diagonalizable. Remember that any anti-hermitian can be diagonalized in $1\times1$ blocks corresponding to its kernel and $2\times2$ blocks. Hence there exists a decomposition $\gH=\oplus_{i\ge 0}\gH_i$, with $\dim(\gH_i)\le 2$ such that 
\begin{itemize}
 \item For all $i$, $\gamma$ and $\alpha J$ stabilize $\gH_i$;
 \item If $\dim(\gH_i)=1$ then $\gamma_{|\gH_i}=\lambda_i\text{Id}_{\gH_i},\alpha J_{|\gH_i}=0$;
 \item If $\dim(\gH_i)=2$ then $\gamma_{|\gH_i}=\left(
 \begin{array}{cc}
  \lambda_i & 0\\
  0 & \lambda_i
 \end{array}\right), 
 \alpha J_{|\gH_i}=\left(
 \begin{array}{cc}
  0 & \alpha_i\\
 -\alpha_i & 0
 \end{array}\right)$ with ${\alpha_i\in\R}$ and $\alpha_i^2=\lambda_i-\lambda_i^2$. 
\end{itemize} 
In particular, $\omega=\otimes_{i\ge 0}\omega_i$ where $\omega_i$ is the quasi-free state on $\F(\gH_i)$ with 1-pdm $\gamma_{|\gH_i}=:\gamma_i$ and pairing matrix $\alpha J_{|\gH_i}$. Let us now prove that for all $i$
 \begin{equation}\label{eq:HFB-pure-ineq_i}
  \omega_i( e^{-\beta\cN_i})=\Det_{\gH_i}\sqrt{1+(e^{-2\beta}-1)\gamma_i},
 \end{equation}
 where $\cN_i$ is the number operator on $\cF(\gH_i)$. First we consider the case $\dim(\gH_i)=1$, and let $f_i\in\gH_i$ be a normalized vector. Then $\F(\gH_i)=\C\oplus\C f_i$ and $\cN_i=a^*(f_i)a(f_i)$ so that $e^{-\beta\cN_i}=1+(e^{-\beta}-1)a^*(f_i)a(f_i)$. Therefore
 $$\omega_i(e^{-\beta\cN_i})=1+(e^{-\beta}-1)\lambda_i.$$
 Since $\gamma^2+\alpha\alpha^*=\gamma$, we have $\alpha_i=0$ if $\dim(\gH_i)=1$ hence $\lambda_i=0$ or $\lambda_i=1$. In both cases we have
 $$\omega_i( e^{-\beta\cN_i})=\Det_{\gH_i}\sqrt{1+(e^{-2\beta}-1)\gamma_i}.$$
 Now suppose $\dim(\gH_i)=2$ and let $(f_i,g_i)$ be an orthonormal basis of $\gH_i$ such that in this basis $\gamma_{|\gH_i}$ and $\alpha J_{|\gH_i}$ have the form given above. Then $\F(\gH_i)=\C\oplus\C f_i\oplus\C g_i\oplus\C f_i\wedge g_i$ and $\cN_i=a^*(f_i)a(f_i)+a^*(g_i)a(g_i)$, so that 
 $$e^{-\beta\cN_i}=(1+(e^{-\beta}-1)a^*(f_i)a(f_i))(1+(e^{-\beta}-1)a^*(g_i)a(g_i)).$$
 We deduce that
 \begin{eqnarray*}
  \omega_i(e^{-\beta\cN_i}) & = & 1+(e^{-\beta}-1)\omega_i(\cN_i)+(e^{-\beta}-1)^2\omega_i(a^*(f_i)a(f_i)a^*(g_i)a(g_i)) \\
 & = & 1+(e^{-\beta}-1)\omega_i(\cN_i)+ (e^{-\beta}-1)^2[\omega_i(a^*(f_i)a(f_i))\omega_i(a^*(g_i)a(g_i))\\
 & & -\omega_i(a^*(f_i)a^*(g_i))\omega_i(a(f_i)a(g_i))+ \omega_i(a^*(f_i)a(g_i))\omega_i(a(f_i)a^*(g_i))] \\
 & = & 1+2(e^{-\beta}-1)\lambda_i+(e^{-\beta}-1)^2(\lambda_i^2+\alpha_i^2) \\
 & = & 1+2(e^{-\beta}-1)\lambda_i+(e^{-\beta}-1)^2\lambda_i \\
 & = & 1+(e^{-2\beta}-1)\lambda_i=\Det_{\gH_i}\sqrt{1+(e^{-2\beta}-1)\gamma_i},
 \end{eqnarray*}
where in the second equality we used Wick's relation for $\omega_i$. The equality \eqref{eq:HFBpureequality} then follows by taking the product of the relations \eqref{eq:HFB-pure-ineq_i}.  Putting aside the indices $i$ such that $\lambda_i=1$, we obtain
$$\omega(e^{-\beta\cN})=e^{-\beta k_0}\prod_{i\in\N}(1+(e^{-2\beta}-1)\lambda_i),$$
where $k_0:=\dim\ker(\gamma-1)$.  Identifying the coefficient of $e^{-\beta k}$ in both sides, as in the proof of Proposition \ref{prop:HF}, we find that $\omega(\cP_k)=0$ for all $k<k_0$, and that
\begin{eqnarray*}
 \omega(\cP_k) & = & \sum_{\substack{I\subset\N \\ \#I=\ell}}\prod_{i\in I}\lambda_i\prod_{j\notin I}(1-\lambda_j)\\
 & \le & \frac{e^\ell}{\ell !}\left(\frac{\tr(\gamma)-k_0}{2}\right)^\ell e^{-\frac{\tr(\gamma)-k_0}{2}} \le \frac{e^{k/2}}{\ell!}\left(\frac{\tr(\gamma)}{2}\right)^\ell e^{-\frac{\tr(\gamma)}{2}},
\end{eqnarray*}
for $k=k_0+2\ell$ with $\ell\ge0$. This concludes the proof of Proposition \ref{prop:HFBpure}.
\end{proof}

\subsection{Mixed Hartree-Fock-Bogoliubov case}\label{sec:HFBmixed}

In the most general case of a mixed HFB state, we cannot apply the same strategy as in the previous cases, i.e. identify $\omega$ as a product of states living on smaller  dimensional spaces. Indeed, $\gamma$ and $\alpha J$ can have no common stable finite-dimensional subspaces. However, we can still prove an estimate on the vacuum expectation.

\begin{proposition}[Vacuum Expectation of a Mixed HFB State]\label{prop:HFBmixed}
 Let $\omega$ a quasi-free state with $\tr(\gamma)<+\ii$. Then we have the following estimate
  \begin{equation}
  \boxed{
   \omega(\proj{\Omega})\le e^{-a\tr(\gamma)}
  }
 \end{equation}
where  $a=\max_{\beta\ge0}\frac{3(e^{\beta}-1)}{4e^{3\beta}+7e^{\beta}-8}=0.0941248\ldots$.
\end{proposition}

\begin{remark}
 We believe that for any sequence of quasi-free state $(\omega_n)_n$ with finite number of particles, we have $\omega_n(e^{-\beta\cN})\to0$ as $\omega_n(\cN)\to+\ii$ for all $\beta>0$, as it was the case in the previous sections. However, our method does not provide this result. 
\end{remark}

\begin{proof}
  Let $\omega$ be a quasi-free state with $\tr(\gamma)<+\ii$. According to \cite[Eq. (2b.26)]{BacLieSol-94}, there exists a Bogoliubov map $\cV:\gH\oplus\gK\to\gH\oplus\gK$, i.e. a unitary operator of the form 
  \begin{equation}\label{bogo}
  \cV=\left(
  \begin{array}{cc}
  U & J^* V J^*\\
  V &  JUJ^*
  \end{array}\right).
  \end{equation}
  with $\tr(V^*V)<+\ii$, such that 
 \begin{equation}\label{eq:GammaU}
  \cV\Gamma(\gamma,\alpha)\cV^*=\Gamma(D,0)
 \end{equation}
 with $D=\text{diag}(\lambda_i)_{i\ge 0}$, $\sum_i\lambda_i<+\ii$, and  $0\le\lambda_i\le1/2$ for all $i$. Let $\omega'$ be the unique HF state associated with $\Gamma(D,0)$ given by Proposition \ref{BacLieSol}. Let us also denote by $\U$ the unitary operator on $\cF$ lifting $\cV$ \cite[Theorem 9.5]{Solovej-07}. Then
  $$\omega(\proj{\Omega})=\omega'(\U|\Omega\rangle\langle\Omega|\U^*)=\omega'(|\Omega_\U\rangle\langle\Omega_\U|),$$
  with $|\Omega_\U\rangle=\U|\Omega\rangle$. We now estimate $|\Omega_\U\rangle\langle\Omega_\U|$ by $e^{-\beta\cN}$ for any $\beta\ge0$. Thus, let us fix $\beta\ge0$. At the end, we will optimize over $\beta$. By \cite[Eq. (67)]{Solovej-07} we can write
  $$|\Omega_\U\rangle=\prod_{i=-K}^{-1}a^*(\eta_i)\prod_{i\ge0}(\alpha_i-\beta_i a^*(\eta_{2i})a^*(\eta_{2i+1}))|\Omega\rangle,$$
  where $(\eta_i)_{i\in\Z}$ is an orthonormal basis in $\gH$, $(\eta_i)_{-K\le i\le-1}$ are eigenvectors of $V^*V$ for the eigenvalue 1, $\alpha_i^2+\beta_i^2=1$ for all $i\ge0$, and $(\beta_i^2)_{i\ge0}$ are the eigenvalues of $V^*V$ strictly between 0 and 1, which are all of multiplicity 2.  We interpret this equality by saying that 
  $$|\Omega_\U\rangle=T^*(\otimes_i|\psi_i\rangle)\in T^*(\otimes_i\F(\gH_i))=\F,$$
  where $\gH_i=\C\eta_i$, $|\psi_i\rangle=\eta_i$ for $-K\le i\le -1 $, $\gH_i=\C\eta_{2i+1}\oplus\C\eta_{2i}$, $|\psi_i\rangle=\alpha_i-\beta_i\eta_{2i}\wedge\eta_{2i+1}$ for $i\ge0$, and $\gH_i=\C\eta_i$, $|\psi_i\rangle=\Omega_i$ (the vacuum in $\F(\gH_i)$) if $i< -K$. Recall that the operator $T$ is the unitary transformation between $\cF(\oplus_i\gH_i)$ and $\otimes_i\cF(\gH_i)$ defined in Appendix \ref{prodstates} . If $-K\le i\le -1$, then $\F(\gH_i)=\C\oplus\C\eta_i$ and the matrix of $|\psi_i\rangle\langle\psi_i|$ in the basis $(1,\eta_i)$ can be dominated by
  $$|\psi_i\rangle\langle\psi_i|=\left(
  \begin{array}{cc}
   0 & 0\\
   0 & 1
  \end{array}\right)
  \le e^\beta\left(
  \begin{array}{cc}
   1 & 0\\
   0 & e^{-\beta}
  \end{array}\right)
  =e^{\beta-\beta\cN_i}\le e^{\frac{e^{2\beta}-1}{2}-\beta\cN_i},$$
  where as usual $\cN_i$ is the number operator on $\F(\gH_i)$. If $i<-K$ then $\F(\gH_i)=\C\oplus\C\eta_i$ and we have
  $$|\psi_i\rangle\langle\psi_i|=\left(
  \begin{array}{cc}
   1 & 0\\
   0 & 0
  \end{array}\right)
  \le \left(
  \begin{array}{cc}
   1 & 0\\
   0 & e^{-\beta}
  \end{array}\right)
  =e^{-\beta\cN_i}.$$
  Now if $i\ge0$ then $\F(\gH_i)=\C\oplus\C\eta_{2i}\oplus\C\eta_{2i+1}\oplus\C\eta_{2i}\wedge\eta_{2i+1}$. In the basis $(1,\eta_{2i},\eta_{2i+1},\eta_{2i}\wedge\eta_{2i+1})$ the matrix of $|\psi_i\rangle\langle\psi_i|$ can be dominated by
 \begin{eqnarray*}
  |\psi_i\rangle\langle\psi_i|  =  \left(
  \begin{array}{cccc}
   \alpha_i^2 & 0 & 0 & -\alpha_i\beta_i\\
   0 & 0 & 0 & 0\\
   0 & 0 & 0 & 0\\
   -\alpha_i\beta_i & 0 & 0 & \beta_i^2
  \end{array}\right) & \le & e^{(e^{2\beta}-1)\beta_i^2}  \left(
  \begin{array}{cccc}
   1 & 0 & 0 & 0\\
   0 & e^{-\beta} & 0 & 0\\
   0 & 0 & e^{-\beta} & 0\\
   0 & 0 & 0 & e^{-2\beta}
  \end{array}\right)\\ 
  & = & e^{(e^{2\beta}-1)\beta_i^2-\beta\cN_i},
 \end{eqnarray*}
 since with
 $$A=\left( \begin{array}{cc}
             \alpha_i^2 & -\alpha_i\beta_i\\
	    -\alpha_i\beta_i & \beta_i^2
            \end{array}
\right), \qquad
B=\left(\begin{array}{cc}
         1 & 0\\
	 0 & e^{-2\beta}
        \end{array}
 \right)$$
 we have $B^{-1/2}AB^{-1/2}\le (1+(e^{2\beta}-1)\beta_i^2)\text{Id}_{\C^2}\le e^{(e^{2\beta}-1)\beta_i^2}\text{Id}_{\C^2}$. We know that
  $$\tr(V^*V)=\sum_{i=-K}^{-1}1+\sum_{i\ge0}2\beta_i^2,$$
  so that
  \begin{eqnarray*}
  |\Omega_\U\rangle\langle\Omega_\U| & = & T^*\left[\bigotimes_{i\in\Z}|\psi_i\rangle\langle\psi_i|\right]T\\
   & \le & T^*\left[\bigotimes_{i=-K}^{-1}e^{\frac{e^{2\beta}-1}{2}-\beta\cN_i}\otimes\bigotimes_{i\ge0}e^{\frac{e^{2\beta}-1}{2}2\beta_i^2-\beta\cN_i}\otimes\bigotimes_{i<-K}e^{-\beta\cN_i}\right]T\\
   & = & \exp\left(\frac{e^{2\beta}-1}{2}\left(\sum_{i=-K}^{-1}1+\sum_{i\ge0}2\beta_i^2\right)-\beta\sum_i\cN_i\right)\\
   & = & e^{\frac{e^{2\beta}-1}{2}\tr(V^*V)}e^{-\beta\cN}.
  \end{eqnarray*}
  Here we have used that for any operators $A,B,C,D,$ such that $0\le A\le B$ and $0\le C\le D$ it holds $0\le A\otimes C\le B\otimes D$. We obtain the estimate
  \begin{equation}\label{firstestimate}
   \omega(\proj{\Omega})\le e^{\frac{e^{2\beta}-1}{2}\tr(V^*V)}\omega'(e^{-\beta\cN})\le e^{\frac{e^{2\beta}-1}{2}\tr(V^*V)+(e^{-\beta}-1)\tr(D)},
  \end{equation}
  where we have used that $\omega'(e^{-\beta\cN})\le e^{(e^{-\beta}-1)\omega'(\cN)}$, which is a consequence of the equality \eqref{eq:HFequality} applied to the HF state $\omega'$. Unfortunately, the estimate (\ref{firstestimate}) is not good enough because the constant $(e^{2\beta}-1)/2$ in front of $\tr(V^*V)$ is positive, while $\tr(V^*V)$ represents the number of particles of $|\Omega_\U\rangle$. We will thus get another estimate by exchanging the roles of $\omega'$ and $|\Omega_\U\rangle$. The idea is to see $|\Omega_\U\rangle$ as the state and $\omega'$ as the observable. Since $\ker(D-1)=\{0\}$, it is well-known that $\omega'$ is a normal state with density matrix $G=Z^{-1}\Upsilon(M)$ where $M=\frac{D}{1-D}$, $Z=\tr_\F(\Upsilon(M))$, and \begin{equation}\label{Upsilon}
    \Upsilon(M):=\oplus_{N\ge0}M^{\otimes N}:\cF\to\cF                                                                                                                                                                                                                                                                                                                                                                                                                                                                                                                                                                                                                                                                                                                                                                                                                                                                                                                                                                                                                                                                                                                                                                                                                                                                                                                                                                                                                                                                                                                                                        \end{equation}
 (see for instance \cite[Proposition 6.6 (1)]{JakOgaPauPil-11}). Hence we can write
  $$\omega(\proj{\Omega})=\omega'(|\Omega_\U\rangle\langle\Omega_\U|)=\tr(Z^{-1}\Upsilon(M)|\Omega_\U\rangle\langle\Omega_\U|)=\omega_\U(Z^{-1}\Upsilon(M)),$$
  where $\omega_\U$ is the pure state associated with the vector $|\Omega_\U\rangle$. We know that $\omega_\U(\cN)=\tr(V^*V)$. Hence if we can dominate $Z^{-1}\Upsilon(M)$ by $e^{-\beta'\cN}$ for a certain $\beta'\ge0$, we will get another estimate on $\omega(\proj{\Omega})$ by applying equality \eqref{eq:HFBpureequality} to the pure state $\omega_\U$. Let $(\mu_i)_i$ the eigenvalues of $M$, $\mu_i=\frac{\lambda_i}{1-\lambda_i}$. The spectrum of $\Upsilon(M)$ on $\gH^k$ is
  $$\sigma(\Upsilon(M)_{|\gH^k})=\left\{\prod_{i\in I}\mu_i,\quad I\subset\N,\quad\# I=k\right\},$$
  so that for $I\subset\N$ of cardinal $k$ we have
  $$Z^{-1}\prod_{i\in I}\mu_i=\prod_{i\ge 0}\frac{1}{1+\mu_i}\prod_{i\in I}\mu_i=\prod_{i\ge 0}(1-\lambda_i)\prod_{i\in I}\frac{\lambda_i}{1-\lambda_i}=\prod_{i\notin I}(1-\lambda_i)\prod_{i\in I}\lambda_i.$$
  Using $0\le\lambda_i\le 1/2$ for all $i$, we finally get $Z^{-1}\Upsilon(M)_{|\gH^k}\le 1/2^k$ so that 
  $$Z^{-1}\Upsilon(M)\le e^{-(\ln 2)\cN}.$$
  Equality \eqref{eq:HFBpureequality} now implies that $\omega_\U(e^{-\beta'\cN})\le e^{\frac{e^{-2\beta'}-1}{2}\omega_\U(\cN)}$ for all $\beta'\ge0$. Choosing $\beta'=\ln 2$, one gets
  \begin{equation}\label{secondestimate}
   \omega(\proj{\Omega})\le\omega_\U(e^{-(\ln 2)\cN})\le e^{-\frac{3}{8}\tr(V^*V)}.
  \end{equation}
  Interpolating the inequalities (\ref{firstestimate}) and (\ref{secondestimate}) we get
  $$\omega(\proj{\Omega})\le e^{\theta(e^{-\beta}-1)\tr(D)+\left[\theta\frac{e^{2\beta}-1}{2}-(1-\theta)\frac{3}{8}\right]\tr(V^*V)},$$
  for all $\beta\ge0$ and $0\le\theta\le1$. We choose $\theta$ such that the coefficients before $\tr(D)$ and $\tr(V^*V)$ are equal since we have $\tr(V^*V)+\tr(D)\ge\tr(U^*DU)+\tr(V^*(1-D)V)=\tr(\gamma)$ by \eqref{eq:GammaU}, using $UU^*\le 1$ and $1-D\le 1$. We thus choose
  $$\theta=\frac{3}{7-8e^{-\beta}+4e^{2\beta}},$$
 and we obtain
 $$\omega(\proj{\Omega})\le e^{\frac{3(e^{\beta}-1)}{8-7e^{\beta}-4e^{3\beta}}\tr(\gamma)},$$
 for all $\beta\ge0$. Optimizing the coefficient before $\tr(\gamma)$, we get the desired estimate with $\beta\simeq0.36443$ and $\theta\simeq0.308194$.
  \end{proof}

\subsection{Bogoliubov-Dirac-Fock case}\label{sec:BDFstates} In this section we introduce the correct setup for studying electron-positron pair creation. Let $\gH$ be a Hilbert space and $\Pi$ be an orthogonal projection on $\gH$. We also need an anti-unitary operator $J$ as in Section \ref{sec:notation}, with the additional assumptions that $\gK=\gH$ (i.e. $J$ maps $\gH$ to $\gH$) and that $J\Pi J^*=\Pi$ or $J\Pi J^*=1-\Pi$. The particle/anti-particle spaces are given by $\gH_+=(1-\Pi)\gH$ and $\gH_-=J \Pi\gH$. Notice that $\gH_-=\Pi\gH$ or $\gH_-=(1-\Pi)\gH$. In the context given by Section \ref{sec:BDF-notations}, we have $\gH=\gH_\Lambda$, $\Pi=P^0_-$ and $J=i\beta\alpha_2\CJ$ the charge conjugation operator on $\gH_\Lambda$ (i.e. such that $J(D^0+V)J^*=-(D^0-V)$ for any scalar potential $V$), with $\CJ$ the complex conjugation on $\gH_\Lambda$. With this choice, vectors of $\gH_-$ are interpreted as states with a positive energy relatively to the Hamiltonian with an opposite charge. Hence, they represent positronic states. Notice that this specific $J$ verifies $JP^0_-J^*=P^0_+=1-P^0_-$. In the sequel, we will keep a triplet $(\gH,\Pi,J)$ satisfying the assumptions given above. 

The mathematical description of Bogoliubov-Dirac-Fock states is a special case of the well known Araki-Wyss representation \cite{Ara-64} (see \cite[Section 6.4]{JakOgaPauPil-11} for a review). Let $\cF_0=\cF(\gH_+)\otimes\cF(\gH_-)$. For $f\in\gH_+$ and $g\in\gH_-$ we denote by $a_+^*(f)$ and $a_-^*(g)$ the usual creation operators on $\F(\gH_+)$ and $\F(\gH_-)$, respectively. We now define the ``creation operator'' on $\F_0$ for all $f\in\gH$ by
$$\psi^*(f):=a_+^*((1-\Pi)f)\otimes \text{Id}_{\F(\gH_-)}+\Upsilon\left(-\text{Id}_{\F(\gH_+)}\right)\otimes a_-(J \Pi f),$$
where $\Upsilon$ is the operation defined by Equation \eqref{Upsilon}. The operators $(\psi^*(f))_f$ are not exactly the usual creation operators in the full Fock space since they create a particle in the state $(1-\Pi)f$, and at the same time they annihilate a anti-particle in the state $J \Pi f$, according to the ``particle-hole'' picture of Dirac's theory. However, they still satisfy the CAR thanks to the ``twist'' $\Upsilon\left(-\text{Id}_{\F(\gH_+)}\right)$ on the second term. A \emph{Bogoliubov-Dirac-Fock} (BDF) state is, by definition, a quasi-free state $\omega$ on the $C^*$-algebra $\gA_0\subset\B(\F_0)$ generated by the $(\psi(f))_{f\in\gH}$. We define the normal ordering $:\hspace{-0.1cm}\psi^*(f)\psi(g)\hspace{-0.1cm}:$ of the operator $\psi^*(f)\psi(g)$ by
\begin{multline*}
:\hspace{-0.1cm}\psi^*(f)\psi(g)\hspace{-0.1cm}:\: =a^*_+(f_+)a_+(g_+)\otimes\text{Id}_{\F(\gH_-)}+a^*_+(f_+)\Upsilon\left(-\text{Id}_{\F(\gH_+)}\right)\otimes a^*_-(J g_-)\\
 +\Upsilon\left(-\text{Id}_{\F(\gH_+)}\right)a_+(g_+)\otimes a_-(J f_-)-\text{Id}_{\F(\gH_+)}\otimes a^*_-(J g_-)a_-(J f_-),
\end{multline*}
where $h_+=(1-\Pi)h$ and $h_-=\Pi h$ for all $h$. It corresponds to moving all the creation operators $a^*$ to the left of annihilation operators $a$. For any BDF state $\omega$, we define its renormalized one-particle density matrix $Q:\gH\to\gH$ by
$$\langle g,Q f\rangle=\omega(:\hspace{-0.1cm}\psi^*(f)\psi(g)\hspace{-0.1cm}:).$$
and its pairing matrix $p:\gH\to\gH$ by the usual formula
$$\langle p J f, g\rangle = \omega(\psi^*(f)\psi^*(g)).$$
Recall that we have already defined the 1-pdm $\gamma$ by $\langle g,\gamma f\rangle=\omega(\psi^*(f)\psi(g))$. Therefore, we have the relation $\gamma=\Pi+Q$. If $\cN$ is the number operator on $\F_0$, 
\begin{equation}\label{eq:numberop}
\cN=\sum_i a^*_+(\phi_{i,+})a_+(\phi_{i,+})\otimes\text{Id}_{\F(\gH_-)}+\text{Id}_{\F(\gH_+)}\otimes a^*_-(\phi_{i,-})a_-(\phi_{i,-}), 
\end{equation}
where $(\phi_{i,+})_i$, $(\phi_{i,-})_i$ are orthonormal basis for, respectively, $\gH_+$ and $\gH_-$, then
$$\omega(\cN)=\tr(Q_{++}-Q_{--}),$$
where $Q_{++}:=(1-\Pi)Q(1-\Pi)$ and $Q_{--}=\Pi Q\Pi$. Hence, $\omega(\cN)<+\ii$ is equivalent to having $Q\in\gS_{1,\Pi}(\gH):=\{Q\in\gS_2(\gH),\,Q_{++},Q_{--}\in\gS_1(\gH)\}$. Notice also that as in the HFB case we have $pp^*\le\gamma-\gamma^2=Q_{++}-Q_{--}-Q^2$, thus $p\in\gS_2(\gH)$ as soon as $\omega(\cN)<+\ii$. 
\begin{remark}
 The other natural number operator $\sum_i \psi^*(f_i)\psi(f_i)$ gives the total number of particles in the system, that is also those of the vacuum $\Pi$. If $\dim(\gH_-)=+\ii$, this number is just $+\ii$. However, we only want to count the number of particles relative to the vacuum $\Pi$. That is why the operator $\cN$ is chosen here: It counts the number of ``real'' electrons $\phi_{i,+}$ and the number of ``holes'' $\phi_{i,-}$ in the vacuum. 
\end{remark}
The following proposition in an easy adaptation of the arguments given in \cite[pp. 449--450]{BacBarHelSie-99}. We give the complete proof here to clarify the link between BDF states and HFB states.  
\begin{proposition}[Link between HFB and BDF States]\label{prop:BDF-existence}
 Let $(Q,p)\in\gS_{1,\Pi}(\gH)\times\gS_2(\gH)$ such that 
 $$0\le\Gamma(\Pi+Q,p)=\left(
 \begin{array}{cc}
  \Pi+Q & p\\
  p^* & 1- J(\Pi+Q)J^*
 \end{array}
\right)\le 1$$
as an operator on $\gH\oplus\gH$. Then there exists a unique, normal, BDF state on $\F_0$ having $Q$ as renormalized 1-pdm and $p$ as pairing matrix.
\end{proposition}
\begin{proof}
 We are going to construct an HFB state on $\car[\gH_+\oplus\gH_-]$ using Proposition \ref{BacLieSol}, and then transform it into a BDF state having the desired property via the unitary transformation
$$T:\F(\gH_+\oplus\gH_-)\to\F_0=\F(\gH_+)\otimes\F(\gH_-),$$
defined linearly by its action on each $N$-particle space
\begin{multline*}
T\phi_{i_1,+}\wedge\cdots\wedge\phi_{i_k,+}\wedge\phi_{j_1,-}\wedge\cdots\wedge\phi_{j_{N-k},-}=\\
\left(\phi_{i_1,+}\wedge\cdots\wedge\phi_{i_k,+}\right)\otimes\left(\phi_{j_1,-}\wedge\cdots\wedge\phi_{j_{N-k},-}\right),
\end{multline*}
with the convention $T\Omega=\Omega_+\otimes\Omega_-$. Since $T$ leaves the scalar product invariant and maps an orthonormal basis for $\F(\gH_+\oplus\gH_-)$ onto an orthonormal basis for $\F_0$, it naturally extends to a unique unitary operator on $\cF(\gH_+\oplus\gH_-)$. It induces the following transformation of the CAR:
\begin{equation}\label{AW-CAR}
 T a(\phi_+\oplus\phi_-)T^*=\psi(\phi_+\oplus\phi_-).
\end{equation}
 Define $p_{++}=(1-\Pi)  p (1-\Pi)$, $p_{+-}=(1-\Pi)p\Pi$, etc. Suppose also that we are in the case where $J\Pi J^*=1-\Pi$, and write $\gH^0_+=(1-\Pi)\gH=\gH_+$, $\gH^0_-=\Pi\gH$. We now introduce
 $$\gamma_0:=\left( 
 \begin{array}{cc}
  Q_{++} & p_{++}\\
  p_{++}^* & -J Q_{--}J^*
 \end{array}
\right),\quad \alpha_0:=\left(
\begin{array}{cc}
 p_{+-} & Q_{+-}\\
 -J Q_{-+}J^* & p_{-+}^*
\end{array}
\right),
$$
where $\gamma_0:\gH_+^0\oplus\gH_+^0\to\gH_+^0\oplus\gH_+^0$ and $\alpha_0:\gH_-^0\oplus\gH_-^0\to\gH_+^0\oplus\gH_+^0$. Let us define $\cJ:\gH_+^0\oplus\gH_+^0\to\gH_-^0\oplus\gH_-^0$ by
$$ \cJ=\left(
\begin{array}{cc}
 J & 0\\
 0 & J^*
\end{array}
\right),
$$
and the operator $\Gamma_0$ on $(\gH_+^0\oplus\gH_+^0)\oplus(\gH_-^0\oplus\gH_-^0)$ by
$$\Gamma_0=\Gamma(\gamma_0,\alpha_0)=\left(
\begin{array}{cc}
 \gamma_0 & \alpha_0\\
 \alpha_0^* & 1-\cJ^*\gamma_0\cJ
\end{array}
\right).$$
We now show that there exists a HFB state $\omega_0$ on $\car[\gH_+^0\oplus\gH_+^0]$ having $\Gamma_0$ as density matrix. Hence we prove that $0\le\Gamma_0\le 1$. Let us first write the block decomposition of $\Gamma(\Pi+Q,p)$ as an operator on $\gH_+^0\oplus\gH_-^0\oplus\gH_+^0\oplus\gH_-^0=\gH\oplus\gH$
$$\Gamma(\Pi+Q,p)=\left(
\begin{array}{cccc}
 Q_{++} & Q_{+-} & p_{++} & p_{+-}\\
 Q_{-+} & 1+Q_{--} & p_{-+} & p_{--}\\
 p^*_{++} & p^*_{-+} & -JQ_{--}J^* & -JQ_{-+}J^*\\
 p^*{+-} & p^*_{--} & -JQ_{+-}J^* & 1-JQ_{++}J^*
\end{array}\right),
$$
where we have used that $J\gH^0_\pm=\gH^0_\mp$ to write the lower right block. Let us consider the unitary operator $W:\gH_+^0\oplus\gH_+^0\oplus\gH_-^0\oplus\gH_-^0\to\gH_+^0\oplus\gH_-^0\oplus\gH_+^0\oplus\gH_-^0$ whose matrix is
$$
W=\left(
\begin{array}{cccc}
 1 & 0  & 0 & 0\\
 0 & 0 & 1 & 0\\
 0 & 0 & 0 & 1\\
 0 & 1 & 0 & 0\\
\end{array}\right).
$$
Then, as it was noticed in \cite{BacBarHelSie-99}, we have the relation
$$\Gamma_0=W^*\Gamma(\Pi+Q,p) W.$$
By assumption, we have $0\le\Gamma(\Pi+Q,p)\le 1$, so that $0\le\Gamma_0\le 1$ as well. Now we have $\tr(\gamma_0)=\tr(Q_{++}-Q_{--})<\ii$ hence by Proposition \ref{BacLieSol} with $\gH=\gH^0_+\oplus\gH^0_+=\gH_+\oplus\gH_-$ and $\gK=\gH^0_-\oplus\gH^0_-$, there exists a unique, normal,  HFB state $\omega_0$ on $\F(\gH_+\oplus\gH_-)$ with finite number of particles having $\Gamma_0$ as density matrix. We define a state on $\F_0$ via the unitary operator $T$ by
$$\omega(A)=\omega_0(T^*AT),\quad\forall A\in\B(\F_0).$$ 
By (\ref{AW-CAR}), $\omega$ is the quasi-free state with renormalized 1-pdm $Q$ and pairing matrix $p$. Furthermore, $\omega$ is obviously normal since $\omega_0$ is normal.  This concludes the proof of Proposition \ref{prop:BDF-existence}, in the case where $J\Pi J^*=1-\Pi$. If $J\Pi J^*=\Pi$, the proof is the same with $W:\gH_+^0\oplus\gH_-^0\oplus\gH_+^0\oplus\gH_-^0\to\gH_+^0\oplus\gH_-^0\oplus\gH_+^0\oplus\gH_-^0$ whose matrix is 
$$
W=\left(
\begin{array}{cccc}
 1 & 0  & 0 & 0\\
 0 & 0 & 0 & 1\\
 0 & 0 & 1 & 0\\
 0 & 1 & 0 & 0\\
\end{array}\right),
$$
which is the same as in \cite{BacBarHelSie-99} where $J$ was the complex conjugation. With this choice of $W$, $W^*\Gamma(\Pi+Q,p) W$ is the density matrix of a HFB state with $\gH=\gH^0_+\oplus\gH^0_-=\gK$.
\end{proof}

\begin{corollary}[Vacuum Expectation of a BDF State]\label{coro:estimate-BDF}
 Let $\omega$ be a BDF state with renormalized one-particle density matrix $Q$ such that $\tr(Q_{++}-Q_{--})<+\ii$. Then we have the estimate
 \begin{equation}\label{eq:BDF-ineq}
  \boxed{
 \omega(\proj{\Omega_0})\le e^{-a\tr(Q_{++}-Q_{--})},
 }
 \end{equation}
 where $a$ is the same constant as in Proposition \ref{prop:HFBmixed} and $\Omega_0$ is the vacuum in $\cF_0$.
\end{corollary}
\begin{proof}
 Use Proposition \ref{prop:HFBmixed} to the HFB mixed state $\omega_0$ constructed in the proof of Proposition \ref{prop:BDF-existence}.
\end{proof}

\begin{remark}
 In \cite{HaiLewSol-07}, the rBDF energy \eqref{rBDF-energy} is derived by evaluating the QED Hamiltonian on BDF states. As usual in mean-field theories, this energy only depends on $(Q,p)$. However, since the interaction between the particules is repulsive, any minimizer of the rBDF energy has $p=0$. This explains why we can take $p=0$ for the polarized vacuum in the proof of Theorem \ref{th:main}.
\end{remark}

\section{Asymptotics of the Polarized Vacuum Energy in strong external fields}\label{sec:asymptotics}

\subsection{Main Results}

In this section, we study the asymptotics of the reduced BDF ground state energy $E(Z)$ for large $Z$, where we recall that $E(Z) = \inf \{ \Erbdf{Z\nu}{Q}, Q\in\cK \}$, with $\cK$ defined by Equation \eqref{eq:cK}.

\begin{proposition}[Upper Bound]\label{prop:limsup}
  Let $\alpha>0$ and $\Lambda>0$. Let $\nu\in\cC\cap L^1(\R^3,\R)$ be such that $\int|\nu|=1$ and $q=\int\nu\neq0$. Then 
\begin{equation}
\boxed{
 \limsup_{Z\to+\ii}\frac{E(Z)}{Z^{5/3}}\le-c_1\alpha\Lambda |q|^{5/3},
}
\end{equation}
 where $c_1=2^{-23/3}3^{1/3}\pi^{-4/3}=0.001543\ldots$. 
\end{proposition}
To estimate the convergence rate of $E(Z)Z^{-5/3}$, we need a further assumption on the decay of $\nu$ at infinity.
\begin{proposition}[Convergence Rate]\label{prop:speed}
  Let $\alpha>0$ and $\Lambda>0$. Let $\nu\in\cC\cap L^1(\R^3,\R)$ be such that $\int|x||\nu(x)|\d{x}<+\ii$, $\int|\nu|=1$ and $q=\int\nu\neq0$. Then there exists a constant $\tilde{Z_1}=\tilde{Z_1}(\Lambda,\alpha,|q|,\||x|\nu\|_1)>0$ such that 
 \begin{equation}\label{eq:speed}
  \forall Z> \tilde{Z_1},\qquad E(Z)\le-\frac{c_1}{2}\alpha\Lambda |q|^{5/3}Z^{5/3},
 \end{equation}
  where $c_1$ is defined in Proposition \ref{prop:limsup}.
\end{proposition}
 \begin{remark}
  The constant $\tilde{Z_1}$ behaves as $\Lambda^3$ when $\Lambda\to\ii$, as $\alpha^{-3/2}$ when $\alpha\to0$ and as $|q|^{-4}$ when $q\to0$. It is probably not optimal.
 \end{remark}
We also give a lower bound for $E(Z)$, proving that the power $Z^{5/3}$ is optimal.
\begin{proposition}[Lower Bound]\label{prop:liminf}
 Let $\alpha>0$, $\Lambda>0$, and $\nu\in\cC\cap L^1(\R^3,\R)$ with $\int|\nu|=1$. Then for all $Z>0$ we have
 \begin{equation}
 \boxed{
  E(Z)\ge-c_2\alpha\Lambda Z^{5/3},
 }
 \end{equation}
 where $c_2:=2^{-4/3}3^{2/3}\pi^{-1/3}= 0.563626\ldots$.
\end{proposition}
From the asymptotics of $E(Z)$ in Proposition \ref{prop:speed} we can now derive a lower bound on the total number of particles and anti-particles in the polarized vacuum.
\begin{corollary}[Average Particle Number of the Polarized Vacuum]\label{coro:number-estimate}
  Let $\alpha>0$, $\Lambda>0$. Let $\nu\in\cC\cap L^1(\R^3,\R)$ with $\int|\nu|=1$, $\int|x||\nu(x)|\d{x}<+\ii$, and $q=\int\nu\neq0$. Then for any minimizer $Q$ for $E(Z)$ and for all $Z>\tilde{Z_1}$, we have
 \begin{equation}
  \tr(Q_{++}-Q_{--})\ge CZ^{2/3},
 \end{equation}
 where $\tilde{Z_1}$ is the same as in Proposition \ref{prop:speed} and $C$ is a constant independent of $Z$, given in Equation \eqref{eq:cste} below.
\end{corollary}
 
\subsection{Proof of the lower bound}

We first give the proof of the lower bound in Proposition \ref{prop:liminf}, which is easier than the upper bound. 

\begin{lemma}\label{key-lemma}
 For any $Q\in\cK$, $\rho_Q\in L^{\ii}(\R^3)$ and $\|\rho_Q\|_{L^\ii}\le\frac{\Lambda^3}{6\pi^2}$.
\end{lemma}
\begin{proof}
 The proof will only use the fact that any $Q\in\cK$ is a bounded operator on $\gH_\Lambda$. Hence, let $Q\in\cK$ and $V\in L^1(\R^3)\cap L^2(\R^3)$. Since $\rho_Q\in L^2(\R^3)$ we know that
 $$\int\rho_Q V=\tr_0(QV),$$
 where in the trace $V$ is seen as a multiplication operator on $L^2(\R^3,\C^4)$. Let us denote by $\Pi_\Lambda$ the multiplication operator by $\ind_{B(0,\Lambda)}$ in Fourier space. Since $Q$ is an operator on $\gH_\Lambda$, we have $\Pi_\Lambda Q\Pi_\Lambda=Q$. Now assume that $V\ge0$. Then
 \begin{eqnarray*}
  \left|\int\rho_Q V\right| & = & \left|\tr_0\left(Q\Pi_\Lambda V\Pi_\Lambda\right)\right|\\
 & \le & \|\Pi_\Lambda V\Pi_\Lambda\|_{\gS_1}=\|\sqrt{V}\Pi_\Lambda\|_{\gS_2}^2\le(2\pi)^{-3}\int V\times\frac{4}{3}\pi\Lambda^3,
 \end{eqnarray*}
 where we used $\|Q\|\le1$ and the Kato-Seiler-Simon inequality:
 $$\forall p\ge2,\qquad\|f(x)g(p)\|_{\gS_p}\le(2\pi)^{-3/p}\|f\|_{L^p}\|g\|_{L^p}.$$
 Now if $V$ is not necessarily non-negative, we split $V=V_+-V_-$ with $V_+=\max(V,0)$ and $V_-=\max(-V,0)$. Then we apply the previous bound twice to obtain
 $$\left|\int\rho_Q V\right|\le\left|\int\rho_Q V_+\right|+\left|\int\rho_Q V_-\right|\le\frac{\Lambda^3}{6\pi^2}\left(\int V_+ +\int V_-\right)=\frac{\Lambda^3}{6\pi^2}\|V\|_{L^1}.$$
 By the density of $L^1\cap L^2$ in $L^1$ and by the fact that $(L^1)'\simeq L^\ii$, we get the result. 
\end{proof}
 Lemma \ref{key-lemma} is crucial to understand the $Z^{5/3}$ behaviour of $E(Z)$. Indeed, an easy lower bound to $\cE_{\text{rBDF}}^{Z\nu}$ is $\cE_{\text{rBDF}}^{Z\nu}(Q)\ge -\alpha Z^2 D(\nu,\nu)/2$ for all $Q$, using the positivity of the kinetic energy and completing the square in the other terms. One may think that this lower bound would be attained by a $Q$ such that $\rho_Q\simeq Z\nu$, i.e. by a state which density of charge compensates the external field. However, Lemma \ref{key-lemma} implies that such a state cannot exist in $\cK$, precisely because of the cut-off $\Lambda$. In other words, the vacuum cannot ``follow'' the external field when the field is too strong.
\begin{proof}[Proof of Proposition \ref{prop:liminf}]. Let $Q\in\cK$, then 
 $$\Erbdf{Z\nu}{Q}\ge\widetilde{\cE}^{Z\nu}(\rho_Q):=-\alpha D(\rho_Q,Z\nu)+\frac{\alpha}{2}D(\rho_Q,\rho_Q).$$
We also know that for all $Q\in\cK$, one has $\|\rho_Q\|_{L^\ii}\le\frac{\Lambda^3}{6\pi^2}=:\delta$, so that 
$$E(Z)\ge\widetilde{E}(Z,\delta):=\inf\left\{\widetilde{\cE}^{Z\nu}(\rho),\:\rho\in\cC\cap L^\ii(\R^3), \|\rho\|_{L^\ii}\le \delta\right\}.$$
The variational problem $\widetilde{E}(Z,\delta)$ has the scaling property
$$\widetilde{E}(Z,\delta)=Z^2\widetilde{E}(1,\delta/Z).$$
Define $\epsilon=\delta/Z$. We will now show that
$$\widetilde{E}(1,\epsilon)\ge -\frac{3}{4}(8\pi\epsilon)^{1/3},$$
which then implies Proposition \ref{prop:liminf}. Let $\rho$ be a trial state for $\tilde{E}(1,\epsilon)$. Let $R>0$ and write $|\cdot|^{-1}=V_1+V_\ii$ with $V_1:=\ind_{|x|\le R}|\cdot|^{-1}\in L^1(\R^3)$ and $V_\ii:=\ind_{|x|\ge R}|\cdot|^{-1}\in L^\ii(\R^3)$. On the one hand
$$\left|\int_{\R^3}\rho(\nu\star V_1)\right|\le\epsilon\|\nu\star V_1\|_{L^1}\le\epsilon\|\nu\|_{L^1}\|V_1\|_{L^1}=2\pi\epsilon R^2,$$
where in the last inequality we used Young's inequality and $\int|\nu|=1$. On the other hand,
$$\int_{\R^3}\rho(\nu\star V_\ii)=\int_{\R^3}\rho(x)\left(\int_{|y|\ge R}\frac{\nu(x-y)}{|y|}\text{d}y\right)\text{d}x=\int_{|y|\ge R}\frac{\widetilde{\rho}(y)}{|y|}\text{d}y,$$
where $\widetilde{\rho}:=\tilde{\nu}\star\rho$ and $\tilde{\nu}(x)=\nu(-x)$ for all $x\in\R^3$. Since $\int_{|y|\ge R}\frac{\tilde{\rho}(y)}{|y|}=D(\tilde{\rho},f)$ with $f=(4\pi)^{-1}\delta_{|x|=R}$, as in \cite[Proof of Theorem II.3]{LieSim-77b}, we have
$$-\int_{|y|\ge R}\frac{\widetilde{\rho}(y)}{|y|}\text{d}y+\frac{1}{2}D(\widetilde{\rho},\widetilde{\rho})\ge-\frac{1}{2}D(f,f)=-\frac{1}{2R}.$$
Notice that 
$$D(\widetilde{\rho},\widetilde{\rho})=4\pi(2\pi)^3\int_{\R^3}\frac{|\widehat{\rho}(k)|^2|\widehat{\nu}(k)|^2}{|k|^2}\text{d}k\le D(\rho,\rho),$$
since $|\widehat{\nu}(k)|\le(2\pi)^{-3/2}\|\nu\|_{L^1}=(2\pi)^{-3/2}$ for all $k$. Hence,
$$\widetilde{E}(1,\epsilon)\ge-2\pi\alpha\epsilon R^2-\frac{\alpha}{2R}.$$
Optimizing over $R$, one gets the result. This finishes the proof of Proposition \ref{prop:liminf}.
\end{proof}

\subsection{Proof of the upper bound}

Both Propositions \ref{prop:limsup} and \ref{prop:speed} are proved via appropriate trial states. Before turning to the estimates, we start by explaining our choice of the trial states. In the sequel, we will assume that $\int\nu>0$. The case $\int\nu<0$ is treated in the same fashion, except for the choice of the trial state (see Remark \ref{rk:qneg}). We define the operator $Q$ on $\gH_\Lambda$ by its kernel in Fourier space. For all $p,q\in\R^3$, let
 $$\hat{Q}(p,q)=\widehat{\gamma}(p,q)X(p)X(q)^*,$$
with
$$\widehat{\gamma}(p,q)=(2\pi)^{-3/2}\int_{B(0,\Lambda/2)}\text{d}\ell\, g_r(p-\ell)\widehat{F_r}(p-q)g_r(q-\ell)$$
and
 $$X(p)=U(p)\left[\left(1\;0\;0\;0\right)^t\right]$$
 with
 $$U(p)=a_+(p)-a_-(p)\beta\frac{\alp\cdot p}{|p|},\qquad a_\pm(p)=\frac{1}{\sqrt{2}}\sqrt{1\pm\frac{1}{1+|p|^2}}.$$
The operator $U(p)$ is unitary on $\C^4$ for all $p$ and it diagonalizes $D^0(p):=\alp\cdot p+\beta$ as 
 $U(p)D^0(p)U(p)^*=\sqrt{1+|p|^2}\beta$. In the definition of $\hat{\gamma}$, we choose $0<r<\Lambda/2$ (a small number which will eventually tend to zero), $g_r=r^{-3/2}g(\cdot/r)$ with $g\in L^2(\R^3,\R)$, $\int g^2=1$ and $\text{supp}(g)\subset B(0,1)$. We also choose $F_r=F(r\cdot)$ with $F\in L^1(\R^3,\R)$ such that $0\le F(x)\le1$ for all $x$. 

\begin{lemma}
 The operator $Q$ belongs to the variational set $\cK$.
\end{lemma}
\begin{proof}
 First, notice that $Q$ defines a Hilbert-Schmidt operator on $L^2(\R^3,\C^4)$ since for instance
 \begin{multline*}
\|Q(\cdot,\cdot)\|_{L^2(\R^3,\C^4)}^2\le(2\pi)^{-3}\|g_r\|_{L^2}^2\int_{B(0,\Lambda/2)}\text{d}\ell\int_{\R^3}\text{d}p\left( |\widehat{F_r}|^2\star |\tau_\ell g_r|^2\right)(p)\\
  \le(2\pi)^{-3} \|g_r\|_{L^2}^2\int_{B(0,\Lambda/2)}\text{d}\ell \left\| |\widehat{F_r}|^2\star |\tau_\ell g_r|^2 \right\|_{L^1}\\
  \le\frac{\text{vol}(B(0,\Lambda/2))}{(2\pi)^3}\|g_r\|_{L^2}^4\|F_r\|_{L^1}^2 <\ii
 \end{multline*}
by Young's inequality. It is self-adjoint because $\widehat{Q}(p,q)=\overline{\widehat{Q}(q,p)}$ for all $p,q$,
and we have $\supp\widehat{Q}(\cdot,\cdot)\subset (B(0,\Lambda/2)+\supp g_r )^2\subset B(0,\Lambda)^2$ hence $Q$ is an operator with range in $\gH_\Lambda$. Since for all $\phi\in\gH_\Lambda$ we have 
$$\hat{Q\phi}(p)=\left(\int_{\R^3}\hat{\gamma}(p,q)\left\langle X(q),\hat{\phi}(q)\right\rangle_{\C^4}\d{q}\right)X(p),$$
we conclude $Q_{--}=Q_{+-}=Q_{-+}=0$ so that $Q=Q_{++}$. Let $\phi\in\gH_+^0$. Then
\begin{eqnarray*}
 \langle Q\phi,\phi\rangle_{L^2(\R^3,\C^4)} & = & (2\pi)^{-3/2} \int_{B(0,\Lambda/2)}\text{d}\ell\,\left\langle (\tau_\ell g)X\widehat{\phi},\widehat{F}\star((\tau_\ell g)X\widehat{\phi})\right\rangle_{L^2(\R^3,\C^4)}\\
 & = & \int_{B(0,\Lambda/2)}\text{d}\ell\int_{\R^3}\text{d}x\, F(x) \left\|\cF^{-1}\left((\tau_\ell g)X\widehat{\phi}\right)(x)\right\|_{\C^4}^2\\
 & \le & \int_{B(0,\Lambda/2)}\text{d}\ell\int_{\R^3}\text{d}p\,\left\|g(p-\ell)X(p)\widehat{\phi}(p)\right\|_{\C^4}^2\\
 & \le & \|\phi\|_{L^2(\R^3,\C^4)}^2,
\end{eqnarray*}
where we have denoted by $\cF^{-1}$ the inverse Fourier transform. Hence $-P^0_-\le Q\le1-P^0_-$. Finally, $\tr(Q)=\int_{\R^3}\hat{\gamma}(p,p)\|X(p)\|_{\C^4}^2\d{p}\le(2\pi)^{-3}\int F<+\ii$, so that $Q\in\gS_1(\gH_\Lambda)\subset\gS_{1,P^0_-}(\gH_\Lambda)$. 
\end{proof}
Since for any $R\in\cK$ we have the formula
$$\hat{\rho_R}(k)=\frac{1}{(2\pi)^{3/2}}\int_{\R^3}\tr_{\C^4}(\hat{R}(p+k,p))\d{p},\qquad \forall k\in\R^3,$$
the density of $Q$ can be written as $\rho_Q=\rho_1+\rho_2$ where for all $k$,  
$$\hat{\rho_1}(k)=(2\pi)^{-3}V_\Lambda\hat{F_r}(k)g_r\star\tilde{g_r}(-k)$$
 and 
$$\hat{\rho_2}(k)=(2\pi)^{-3}V_\Lambda\hat{F_r}(k)\int_{\R^3}g_r(p)g_r(p+k)\left\langle X(p),(X(p+k)-X(p))\right\rangle_{\C^4}\d{p}$$
with $V_\Lambda:=\text{vol}(B(0,\Lambda/2))$ and $\tilde{g}:=g(-\cdot)$.
\begin{proof}[Proof of Proposition \ref{prop:limsup}]
 We start by estimating the terms giving the $Z^{5/3}$ behaviour. We have 
$$D(\rho_1,Z\nu)=4\pi Z\frac{V_\Lambda}{(2\pi)^3} r^{-2}\int_{B(0,2\Lambda)}\widehat{F}(k)\frac{g\star\widetilde{g}(-k)\overline{\widehat{\nu}(rk)}}{|k|^2}\,\text{d}k$$
and
$$D(\rho_1,\rho_1)=4\pi\left(\frac{V_\Lambda}{(2\pi)^3}\right)^2 r^{-5}\int_{B(0,2\Lambda)}|\widehat{F}(k)|^2\frac{|g\star\widetilde{g}(-k)|^2}{|k|^2}\,\text{d}k.$$
We choose $r$ such that $(2\pi)^{-3/2}qZ\frac{V_\Lambda}{(2\pi)^3} r^{-2}=\left(\frac{V_\Lambda}{(2\pi)^3}\right)^2 r^{-5}$, ie $r=\frac{1}{\sqrt{2\pi}}\left(\frac{qZ}{V_\Lambda}\right)^{-1/3}$. The constraint $r<\Lambda/2$ is equivalent to $qZ>\frac{1}{(2\pi)^{3/2}}\frac{4\pi}{3}$, which is automatically satisfied in the limit $Z\to+\ii$. We will come back to it in the proof of Proposition \ref{prop:speed}. We thus get by the dominated convergence theorem
\begin{multline}\label{eq:lim}
 \lim_{Z\to+\ii}Z^{-5/3}\left(-\alpha D(\rho_1,Z\nu)+\frac{\alpha}{2}D(\rho_1,\rho_1)\right)=2^{-11/6}3^{-1/3}\pi^{-13/6}\alpha\Lambda q^{5/3}\times \\
 \left(-\int_{B(0,2\Lambda)}\frac{\hat{F}(k)g\star\tilde{g}(-k)}{|k|^2}\d{k}+\frac{1}{2}\int_{B(0,2\Lambda)}\frac{|\hat{F}(k)g\star\tilde{g}(-k)|^2}{|k|^2}\d{k}\right).
\end{multline}
We now want to optimize the right side with respect to $g$ and $F$. We choose $F=\ind_{B(0,a)}$ and $g=(3(4\pi)^{-1}b^{-3})^{1/2}\ind_{B(0,b)}$ with $a>0$ and $b\in(0,1)$, and we optimize over $a$ and $b$. Since $g\star\tilde{g}(k)=(\text{vol}(B(0,b)))^{-1}\text{vol}(B(0,b)\cap B(-k,b))$ for all $k$ and since $B(-k/2,b-|k|/2)\subset B(0,b)\cap B(-k,b)$ for all $|k|\le 2b$, we have
$$g\star\widetilde{g}(k)\ge\ind_{B(0,2b)}(k)\left(1-\frac{|k|}{2b}\right)^3.$$
For all $|k|\le 2b$, we also have
$$\widehat{F}(k)\ge\widehat{F}(0)-\frac{\int|x|F}{(2\pi)^{3/2}}2b=\frac{a^3}{\sqrt{2\pi}}\left(\frac{2}{3}-ab\right).$$
Therefore,
\begin{eqnarray*}
 \int_{B(0,2\Lambda)}\frac{\widehat{F}(k)g\star\widetilde{g}(-k)}{|k|^2}\text{d}k & \ge & \frac{a^3}{\sqrt{2\pi}}\left(\frac{2}{3}-ab\right)\int_{B(0,2b)}\left(1-\frac{|k|}{2b}\right)^3\frac{\text{d}k}{|k|^2}\\
 & = & \sqrt{2\pi}a^3 b\left(\frac{2}{3}-ab\right).
\end{eqnarray*}
Then, using $|\widehat{F}(k)|\le 4\pi a^3/(3(2\pi)^{3/2})$ and $|g\star\widetilde{g}(k)|\le1$ for all $k$, we obtain
\begin{equation}\label{eq:maj1}
 \int_{B(0,2\Lambda)}\frac{|\widehat{F}(k)g\star\widetilde{g}(-k)|^2}{|k|^2}\le\frac{16}{9}a^6b.
\end{equation}
Hence for the $Z^{5/3}$ term we find that
\begin{multline*}
 \lim_{Z\to+\ii}Z^{-5/3}\left(-\alpha D(\rho_1,Z\nu)+\frac{\alpha}{2}D(\rho_1,\rho_1)\right)\le2^{-11/6}3^{-1/3}\pi^{-13/6}\alpha\Lambda q^{5/3}\times\\ a^3b\left(\sqrt{2\pi}\left(ab-\frac{2}{3}\right)+\frac{8}{9}a^3\right).
\end{multline*}
Optimizing the right side over all $a>0$ and $b\in(0,1)$ we get the result because
$$\min_{\substack{a>0 \\ b\in(0,1)}}a^3b\left(\sqrt{2\pi}\left(ab-\frac{2}{3}\right)+\frac{8}{9}a^3\right)=-2^{-35/6}3^{2/3}\pi^{5/6}.$$


It now remains to prove that the other terms in $\Erbdf{Z\nu}{Q}$ are of lower order than $Z^{5/3}$. We begin with the kinetic energy, 
$$\tr_0(D^0Q)\le\sqrt{1+\Lambda^2}\tr(Q)=(2\pi)^{-3}r^{-3}\int F=(2\pi)^{-3/2}\frac{\sqrt{1+\Lambda^2}}{\Lambda^3}a^3qZ.$$
For the terms involving $\rho_2$, we first use that for all $p,k$,
$$\|X(p+k)-X(p)\|\le\|U(p+k)-U(p)\|\le\frac{7}{\sqrt{2}}|k|.$$
Consequently, for all $k$,
$$|\hat{\rho_2}(k)|\le\frac{7V_\Lambda}{\sqrt{2}(2\pi)^3}r^{-3}|\hat{F}(k/r)||g\star\tilde{g}(-k/r)||k|.$$
Using this bound together with the estimates leading to \eqref{eq:maj1}, one finds that
$$\left\{
\begin{array}{ccc}
 |D(\rho_2,Z\nu)| & \le & \beta_1\Lambda^2q^{1/3}Z^{4/3},\\
 |D(\rho_1,\rho_2)| & \le & \beta_2\Lambda^2q^{4/3}Z^{4/3},\\
 |D(\rho_2,\rho_2)| & \le & \beta_3\Lambda^3qZ,
\end{array}
\right.
$$
with
$$\left\{
\begin{array}{ccc}
\beta_1 & = & 2^{19/3}3^{-2/3}7\pi^{23/6}a^3b^2,\\
\beta_2 & = & 2^{11/6}3^{-8/3}7\pi^{-7/3}a^6b^2,\\
\beta_3 & = & 2^{-3/2}3^{-4}7^2\pi^{-5/2}a^6b^3. 
\end{array}
\right.
$$
We conclude that
$$\limsup_{Z\to+\ii}\frac{E(Z)}{Z^{5/3}}\le\lim_{Z\to+\ii}\frac{\Erbdf{Z\nu}{Q}}{Z^{5/3}}=-c_1\alpha\Lambda q^{5/3}.$$
\end{proof}
\begin{remark}\label{rk:qneg}
 In the case $\int\nu<0$, the proof is the same except that we take
 $$X(p)=U(p)\left[\left(0\;0\;0\;1\right)^t\right],\qquad\forall p\in\R^3,$$
 and
 $$\hat{Q}(p,q)=-\widehat{\gamma}(p,q)X(p)X(q)^*,\qquad\forall p,q\in\R^3.$$
 The trial state $Q$ now verifies $Q=Q_{--}$, so that $\hat{\rho_Q}$ is locally negative around $0$.
\end{remark}

To estimate the convergence rate of $E(Z)$ towards $-c_1\alpha\Lambda |q|^{5/3}Z^{5/3}$, we will use the first moment $\int|x||\nu|$ to control the convergence of $\hat{\nu}(k)$ to $\hat{\nu}(0)$. 
\begin{proof}[Proof of Proposition \ref{prop:speed}]
 We assume $q>0$ (the case $q<0$ follows from obvious modifications). We split the term $D(\rho_1,Z\nu)$ into
 $$D(\rho_1,Z\nu)=\frac{4\pi |q|Z}{(2\pi)^{3/2}}\int_{B(0,2\Lambda)}\frac{\hat{\rho_1}(k)}{|k|^2}\d{k}+R_1,$$
with 
 $$R_1=4\pi Z\int_{B(0,2\Lambda)}\frac{\hat{\rho_1}(k)\overline{\hat{\nu}(k)-\hat{\nu}(0)}}{|k|^2}\d{k}.$$
We use that $|\hat{\nu}(k)-\hat{\nu}(0)|\le(2\pi)^{-3/2}\||x|\nu\|_1 |k|$ for all $k$ to estimate
$$|R_1|\le 2^{5/3}3^{-5/3}\pi^{-11/6}a^3b^2\||x|\nu\|_1\Lambda^2|q|^{1/3}Z^{4/3}.$$
Hence we have for all $|q|Z>\frac{1}{(2\pi)^{3/2}}\frac{4\pi}{3}$,
\begin{equation}\label{eq:maj2}
 E(Z)\le\Erbdf{Z\nu}{Q}\le-c_1\alpha\Lambda |q|^{5/3}Z^{5/3}(1-B_1Z^{-1/3}-B_2Z^{-2/3}),
\end{equation}
with 
\begin{multline*}
  B_1 = 2^{28/3}3^{-2}\pi^{-1/2}a^3b^2\||x|\nu\|_1\Lambda |q|^{-4/3}+2^{4/3}3^{-1}7\pi^{31/6}a^3b^2\Lambda |q|^{-4/3}\\
 +2^{51/6}3^{-3}7\pi^{-1}a^6b^2\Lambda |q|^{-1/3}
\end{multline*}
and 
$$B_2 = 2^{37/6}3^{-1/3}\pi^{4/3}a^3\frac{\sqrt{1+\Lambda^2}}{\Lambda^4}\alpha^{-1}|q|^{-2/3}+2^{31/6}3^{-11/3}7^2\pi^{-7/6}a^6b^3\Lambda^2|q|^{-2/3}.$$
We furthermore have $1-B_1X-B_2X^2\ge 1/2$ for $0\le X\le X_0$ with
$$X_0=\frac{B_1}{2B_2}\left( \sqrt{1+\frac{2B_2}{B_1^2}}-1\right),$$
therefore Proposition \ref{prop:speed} holds with 
\begin{equation}\label{eq:Z_1}
\tilde{Z_1}:=\max\left(X_0^{-3},\frac{1}{(2\pi)^{3/2}}\frac{4\pi}{3|q|}\right). 
\end{equation}
\end{proof}
\begin{proof}[Proof of Corollary \ref{coro:number-estimate}]
 For any $Q\in\cK$, we have the estimate \cite[Lemma 1]{HaiLewSer-08}
 $$D(\rho_Q,\rho_Q)^{1/2}\le C_\Lambda\|Q\|_{1,P^0_-},$$
with $C_\Lambda<+\ii$. From the constraint $-P^0_-\le Q\le 1-P^0_-$ we also get $Q_{++}-Q_{--}\ge Q^2$, hence defining $X:=\left[\tr(Q_{++}-Q_{--})\right]^{1/2}$ we have
$$2\|Q_{+-}\|_{\gS_2}^2=2\|Q_{-+}\|_{\gS_2}^2=\|Q_{+-}\|_{\gS_2}^2+\|Q_{-+}\|_{\gS_2}^2\le\|Q\|_{\gS_2}^2\le X^2,$$
hence $\|Q\|_{1,P^0_-}\le X^2+\sqrt{2}X$. Using the Cauchy-Schwarz inequality for the Coulomb scalar product we also get that for all $Q\in\cK$
$$D(\rho_Q,Z\nu)\le C_\Lambda Z \|\nu\|_\cC(X^2+\sqrt{2}X).$$
From this estimate and the inequality \eqref{eq:speed}, we can see that for all $Z>\tilde{Z_1}$ and for all minimizer $Q^{\text{vac}}$ for $E(Z)$,
$$
-\alpha C_\Lambda Z \|\nu\|_\cC(X^2+\sqrt{2}X)\le\Erbdf{Z\nu}{Q^{\text{vac}}}=E(Z)\le-\frac{c_1}{2}\alpha\Lambda |q|^{5/3}Z^{5/3}.
$$
Hence,
$$X\ge \frac{1}{\sqrt{2}}\left( \sqrt{1+2a}-1\right),$$
with $a=\frac{c_1\Lambda |q|^{5/3}}{2C_\Lambda\|\nu\|_\cC}Z^{2/3}$, so that
$$\tr(Q_{++}^{\text{vac}}-Q_{--}^{\text{vac}})=X^2\ge\frac{1}{2}\left( \sqrt{1+2a}-1\right)^2\ge CZ^{2/3},$$
where
\begin{equation}\label{eq:cste}
 C=\left(\sqrt{1+\frac{c_1\Lambda |q|^{5/3}}{C_\Lambda\|\nu\|_\cC}\tilde{Z_1}^{2/3}}-1\right)^2\tilde{Z_1}^{-2/3},
\end{equation}
using that $x\mapsto(\sqrt{1+x}-1)/\sqrt{x}$ is increasing on $(0,+\ii)$.
\end{proof}

\section{Proof of Theorem \ref{th:main}}\label{sec:proof}

\begin{proof}[Proof of Theorem \ref{th:main}]
 Let $Q$ be any minimizer for $E(Z)$, and let $\omega$ be the unique BDF state on $\cF(\gH_+)\oplus\cF(\gH_-)$ having $Q$ as its generalized 1-pdm and $p=0$ as its pairing matrix, defined by Proposition \ref{prop:BDF-existence}. Then by Corollary \ref{coro:estimate-BDF}, we have 
 $$p_Z=1-\omega(\proj{\Omega})\ge 1-e^{-a\tr(Q_{++}-Q_{--})}.$$
 Now by Corollary \ref{coro:number-estimate}, we know that for all $Z>\tilde{Z_1}$, 
 $$\tr(Q_{++}-Q_{--})\ge CZ^{2/3}.$$
Thus $p_Z\ge 1-e^{-\kappa Z^{2/3}}$ with $\kappa:=aC$.
\end{proof}

\appendix

\section{Product States}\label{prodstates}

Given a (at most) countable family of separable Hilbert spaces $(\gH_i)_{i\in \N}$ and a family $(\omega_i)_{i\in\N}$ of quasi-free states such that $\omega_i$ is a state on $\F(\gH_i)$ for all $i\in\N$, we want to give a meaning to the product state $\otimes_{i\in \N}\omega_i$ as a quasi-free state on $\F(\oplus_{i\in \N}\gH_i)$ .  We first consider the unitary transformation
$$\begin{array}{cccc}
    T: &  \F\left(\bigoplus_{i\in \N}\gH_i\right) & \longrightarrow & \bigotimes_{i\in \N}\F(\gH_i) \\
     & \bigwedge_{j\in J_{i_1}}\phi_{j,i_1}\wedge\cdots\wedge\bigwedge_{j\in J_{i_k}}\phi_{j,i_k} & \longmapsto &   \bigwedge_{j\in J_{i_1}}\phi_{j,i_1}\otimes\cdots\otimes\bigwedge_{j\in J_{i_k}}\phi_{j,i_k},
    \end{array}
$$
where $i_1 < \cdots < i_k$ are elements of  $\N$ and for each $1\le\ell\le k$, $(\phi_{j,i_\ell})_{j\in\N}$ is an orthonormal basis for $\gH_{i_\ell}$ and $J_{i_\ell}\subset\N$ is finite. We recall the definition of a product state on a tensor product of $C^*$-algebras \cite[Proposition 2.9]{Gui-66}.
\begin{theorem}
 Let $(\gA_i)_{i\in\N}$ be a collection of (unital) $C^*$-algebras and let $\omega_i$ be a state on $\gA_i$ for all $i$. There exists a unique state on $\otimes_i \gA_i$, denoted by $\otimes_i\omega_i$ such that for any $(A_{i_1},\ldots,A_{i_\ell})\in\gA_{i_1}\times\cdots\times\gA_{i_\ell}$ with any indices $i_1<\cdots< i_\ell$, we have $\otimes_i\omega_i\left(A_{i_1}\otimes\cdots\otimes A_{i_\ell}\right)=\prod_{k=1}^\ell\omega_{i_k}(A_{i_k})$.
\end{theorem}
\begin{remark}
 Recall that the tensor product $\otimes_i\gA_i$ is defined as the inductive limit of the $C^*$-algebras $\otimes_{i\in J}\gA_i$ with $J$ finite, or equivalently as the completion (for a certain $*$-norm) of $\otimes_i\gA_i$ seen as a tensor product of unital algebras \cite{Gui-66}.
\end{remark}
  In the particular case where $\gA_i=\car[\gH_i]$ for all $i\in\N$, the tensor product $\otimes_i\gA_i$ is also a CAR algebra generated by the operators
$$\widehat{a}(f_j):=\Upsilon\left(-\text{Id}_{\gH_0}\right)\otimes\cdots\otimes\Upsilon\left(-\text{Id}_{\gH_{j-1}}\right)\otimes a_j(f_j),$$
where $j\in\N$, $f_j\in\gH_j$ and $a_j$ is the annihilation operator on $\F(\gH_j)$ for all $j$. Notice the ``twisting'' operator $\Upsilon\left(-\text{Id}\right)$ on the left used to ensure that the $(\widehat{a}(f_j))_{f_j,j}$ satisfy the CAR. Notice also that we have the relation
$$\widehat{a}(f_j)=T a(f_j) T^*,$$
where $a(f_j)$ is the usual annihilation operator on $\F(\oplus_i\gH_i)$.
\begin{proposition}\label{propprodstates}
 Let $(\omega_i)_{i\in\N}$ be a collection of quasi-free states such that $\omega_i$ is defined on $\car[\gH_i]$ for all $i$. Let $(\gamma_i,\alpha_i)_i$ be the collection of their density matrices, and assume that $\sum_i\tr_{\gH_i}(\gamma_i)<+\ii$. Let $\omega$ be the unique quasi-free state on $\car[\oplus_i\gH_i]$ having $\gamma:=\oplus_i\gamma_i$ as 1-pdm and $\alpha:=\oplus_i\alpha_i$ as pairing matrix. Then we have
 $$\forall A\in\car[\oplus_i\gH_i],\quad \omega(A)=\otimes_i\omega_i\left(TAT^*\right).$$
 As a consequence, the product state $\otimes_i\omega_i$ can be considered as a state on $\car[\oplus_i\gH_i]$.
\end{proposition}
\begin{remark}
 The state $\omega$ is well-defined by Proposition \ref{BacLieSol}, since $\tr(\gamma)=\sum_i\tr_{\gH_i}(\gamma_i)$ is finite.
\end{remark}
\begin{proof}
 We will show that the state $\otimes_i\omega_i$ is quasi-free on the CAR generated by the $(\widehat{a}(f_j))_{f_j,j}$. This proves the result because the state $\otimes_i\omega_i(T\cdot T^*)$ is then quasi-free and one easily shows that it has $\oplus_i\gamma_i$ as 1-pdm and $\oplus_i\alpha_i$ as pairing matrix. Since the density matrices determine uniquely the quasi-free state, we must have $\otimes_i\omega_i(T\cdot T^*)=\omega$. Therefore we only have to show the Wick relation for the state $\otimes_i\omega_i$. Let us first notice that it is enough to prove it for products of the form
 $$\widehat{a}^\sharp\left(f_1^{(i_1)}\right)\cdots\widehat{a}^\sharp\left(f_{2k_1}^{(i_1)}\right)\cdots\widehat{a}^\sharp\left(f_1^{(i_N)}\right)\cdots\widehat{a}^\sharp\left(f_{2k_N}^{(i_N)}\right),$$
with $i_1<\cdots<i_N$, $k_\ell\in\N$, $f_p^{(i_\ell)}\in\gH_{i_\ell}$ for all $1\le p\le 2k_\ell$, $1\le\ell\le N$, and where $\sharp$ means star or no star. This means two things:
\begin{enumerate}
 \item We can restrict to ordered products with respect to the decomposition $\oplus_{i\in\N}\gH_i$: The $2k_1$ first creation/annihilation operators $\widehat{a}^\sharp\left(f_p^{(i_1)}\right)$ with $1\le p\le 2k_1$ all create/annihilate particles belonging to $\gH_{i_1}$, then the following $2k_2$ creation/annihilation $\widehat{a}^\sharp\left(f_p^{(i_2)}\right)$ with $1\le p\le 2k_2$ create/annihilate particles belonging to $\gH_{i_2}$, etc. We can always order in this way any product of $a^\sharp$s because Wick's formula does not depend on the choice of an ordering.\\
 \item We can restrict to products where there is an even number of particles created/annihilated in each space $\gH_{i_\ell}$. If not the case, both sides of Wick relation can easily be shown to vanish.
\end{enumerate}
This being said, let us compute the left-hand side of Wick's formula:
\begin{eqnarray*}
 X & = & \otimes_i\omega_i\left[\widehat{a}^\sharp\left(f_1^{(i_1)}\right)\cdots\widehat{a}^\sharp\left(f_{2k_1}^{(i_1)}\right)\cdots\widehat{a}^\sharp\left(f_1^{(i_N)}\right)\cdots\widehat{a}^\sharp\left(f_{2k_N}^{(i_N)}\right)\right] \\
 & = & \otimes_i\omega_i\left[ \Upsilon\left(-\text{Id}_{\gH_0}\right)^{\sum_{\ell=1}^N 2k_\ell}\otimes\cdots\otimes\Upsilon\left(-\text{Id}_{\gH_{i_1-1}}\right)^{\sum_{\ell=1}^N 2k_\ell}\otimes\right. \\
 & & \left.\otimes a^\sharp\left(f_1^{(i_1)}\right)\cdots a^\sharp\left(f_{2k_1}^{(i_1)}\right)\Upsilon\left(-\text{Id}_{\gH_{i_1}}\right)^{\sum_{\ell=2}^N 2k_\ell}\otimes\cdots\right.\\
& & \left.\qquad\qquad\qquad\qquad\qquad\qquad\cdots\otimes a^\sharp\left(f_1^{(i_N)}\right)\cdots a^\sharp\left(f_{2k_N}^{(i_N)}\right)\right]\\
 & = & \prod_{\ell=1}^N\omega_{i_\ell}\left[a^\sharp\left(f_1^{(i_\ell)}\right)\cdots a^\sharp\left(f_{2k_\ell}^{(i_\ell)}\right)\right]\\
 & =& \prod_{\ell=1}^N\sum_{\pi_\ell\in\widetilde{\S_{2k_\ell}}}(-1)^{\epsilon(\pi_\ell)}\omega_{i_\ell}\left[a^\sharp\left(f_{\pi_\ell(1)}^{(i_\ell)}\right)a^\sharp\left(f_{\pi_\ell(2)}^{(i_\ell)}\right)\right]\times\cdots \\
& & \qquad\qquad\qquad\qquad\qquad\qquad \cdots\times\omega_{i_\ell}\left[a^\sharp\left(f_{\pi_\ell(2k_\ell-1)}^{(i_\ell)}\right)a^\sharp\left(f_{\pi_\ell(2k_\ell)}^{(i_\ell)}\right)\right].
\end{eqnarray*}
Let us reindex the product
\begin{multline*}
 \widehat{a}^\sharp\left(f_1^{(i_1)}\right)\cdots\widehat{a}^\sharp\left(f_{2k_1}^{(i_1)}\right)\cdots\widehat{a}^\sharp\left(f_1^{(i_N)}\right)\cdots\widehat{a}^\sharp\left(f_{2k_N}^{(i_N)}\right)=\\
\widehat{a}^\sharp\left(f_1\right)\cdots\widehat{a}^\sharp\left(f_{2k_1}\right)\widehat{a}^\sharp\left(f_{2k_1+1}\right)\cdots\widehat{a}^\sharp\left(f_{2(k_1+\ldots+k_N)}\right),
\end{multline*}
such that the sum of Wick's formula (with $K:=k_1+\ldots+k_N$),
$$Y=\sum_{\pi\in\widetilde{\S_{2K}}}(-1)^{\epsilon(\pi)}(\otimes_i\omega_i)\left[\widehat{a}^\sharp\left(f_{\pi(1)}\right)\widehat{a}^\sharp\left(f_{\pi(2)}\right)\right]\cdots(\otimes_i\omega_i)\left[\widehat{a}^\sharp\left(f_{\pi(2K -1)}\right)\widehat{a}^\sharp\left(f_{\pi(2K)}\right)\right]$$
 contains only the terms with a permutation $\pi$ leaving invariant every interval of the form $[2k_\ell +1,2k_{\ell+1}]$, corresponding to the Hilbert space $\gH_{i_\ell}$, $1\le\ell\le N$. This is proved by induction on $K$: We introduce
$$W_\rho^K(e_1\ldots e_{2K}):=\sum_{\pi\in\widetilde{\S_{2K}}} (-1)^{\epsilon(\pi)}\rho(e_{\pi(1)}e_{\pi(2)})\ldots\rho(e_{\pi(2K-1)}e_{\pi(2K)}).$$
For $\pi\in\widetilde{\S_{2K}}$, we have $\pi(1)=1$, hence we have the induction formula
$$W_\rho^K(e_1\ldots e_{2K})=\sum_{i=2}^{2K}(-1)^i\rho(e_1e_i)W_\rho^{K-1}(e_2\ldots\widehat{e_i}\ldots e_{2K}).$$
In our case where $e_j=\widehat{a}^\sharp(f_j)$ and $\rho=\otimes_i\omega_i$, we see that $\rho(e_1e_i)\neq0$ if and only if $i\le 2k_1$. Hence, by induction on $K$, the only permutations $\pi\in\widetilde{\S_{2K}}$ giving rise to a non-zero term in $Y$ are those which leave invariant the intervals  $[2k_\ell +1,2k_{\ell+1}]$ with $1\le\ell\le N$. We can thus write
\begin{eqnarray*}
 Y & = &\prod_{\ell=1}^N  \sum_{\pi_\ell\in\widetilde{\S_{2k_\ell}}}(-1)^{\epsilon(\pi_\ell)}(\otimes_i\omega_i)\left[\widehat{a}^\sharp\left(f_{\pi_\ell(2k_\ell-1)}^{(i_\ell)}\right)\widehat{a}^\sharp\left(f_{\pi_\ell(2k_\ell)}^{(i_\ell)}\right)\right]\times\cdots\\
 & & \qquad\qquad\qquad\qquad\qquad\qquad\cdots
 \times(\otimes_i\omega_i)\left[\widehat{a}^\sharp\left(f_{\pi_\ell(1)}^{(i_\ell)}\right)\widehat{a}^\sharp\left(f_{\pi_\ell(2)}^{(i_\ell)}\right)\right]\\
 & = & \prod_{\ell=1}^N  \sum_{\pi_\ell\in\widetilde{\S_{2k_\ell}}} (-1)^{\epsilon(\pi_\ell)}\omega_{i_\ell}\left[a^\sharp\left(f_{\pi_\ell(2k_\ell-1)}^{(i_\ell)}\right)a^\sharp\left(f_{\pi_\ell(2k_\ell)}^{(i_\ell)}\right)\right]\times\cdots\\
 & & \qquad\qquad\qquad\qquad\qquad\qquad \cdots\times\omega_{i_\ell}\left[a^\sharp\left(f_{\pi_\ell(1)}^{(i_\ell)}\right)a^\sharp\left(f_{\pi_\ell(2)}^{(i_\ell)}\right)\right]\\
 & = & X.
\end{eqnarray*}
This proves that the state $\otimes_i\omega_i$ is quasi-free, which concludes the proof.
\end{proof}

%

\end{document}